%% file: mainArxiv.tex
\documentclass[
a4paper
]{article}
\usepackage{authblk}
\usepackage{fullpage}
\usepackage[utf8]{inputenc}
\usepackage{amsmath}
\usepackage{amssymb}
\usepackage{amsthm}
\usepackage{mathtools}
\usepackage{fullpage}
\usepackage[labelfont=bf]{caption}
\renewcommand{\epsilon}{\varepsilon}
\input{macro}
\theoremstyle{plain}
\newtheorem{theorem}{Theorem}[section]
\newtheorem*{theorem*}{Theorem}

\newtheorem*{proposition*}{Proposition}
\newtheorem{lemma}[theorem]{Lemma}
\newtheorem*{lemma*}{Lemma}

\newtheorem*{corollary*}{Corollary}

\newtheorem*{claim*}{Claim}

\newtheorem*{conjecture*}{Conjecture}
\theoremstyle{definition}
\newtheorem{definition}[theorem]{Definition}
\newtheorem*{definition*}{Definition}
\theoremstyle{remark}

\newtheorem*{observation*}{Observation}
\newtheorem{example}[theorem]{Example}
\newtheorem*{example*}{Example}

\newif\ifArxiv
\Arxivtrue
\begin{document}
\title{{Convolution Sum-Product Queries}}
\author[1]{Kyle Deeds}
\author[2,3]{Timo Camillo Merkl}
\author[4]{Dan Suciu}
\affil[1]{Boston University, United States, \texttt{kdeeds@bu.edu}}
\affil[2]{Max Planck Institute for Software Systems, Germany}
\affil[3]{TU Wien, Austria, \texttt{timo.merkl@tuwien.ac.at}}
\affil[4]{University of Washington, United States, \texttt{suciu@cs.washington.edu}}
\date{} 
\maketitle
\begin{abstract}
\input{0-abstract}
\end{abstract}
\input{1-introduction}
\input{2-background}
\input{3-cspq}
\input{4-basic-algorithm}
\input{5-tree-decompositions}
\input{6-query-evaluation}
\input{7-other-fields}

\input{8-conclusion}
	\bibliographystyle{abbrv}
	\bibliography{biblio}
	\appendix
	\onecolumn
\input{AA-3}
\input{AA-4}
\input{AA-5}
\input{AA-6}
\input{AA-7}
\input{AA-linearEquations}
\end{document}

%% file: macro.tex
\usepackage{graphicx}
\usepackage{amsmath}
\usepackage{amsfonts}
\usepackage{subcaption}
\usepackage{footmisc}
\usepackage{algorithm}
\usepackage[noend]{algpseudocode}
\algrenewcommand\algorithmicdo{}
\algrenewcommand\algorithmicthen{}
\usepackage{booktabs}
\usepackage{tikz}
\usepackage{bm}
\usepackage{xspace}
\usepackage{xcolor}
\usepackage[normalem]{ulem}
\usepackage{stmaryrd}
\usepackage{thm-restate}
\usepackage{pifont}
\usepackage{wrapfig}
\usepackage{oubraces}
\usepackage{nicefrac}
\usepackage{stmaryrd}
\allowdisplaybreaks
\renewcommand{\phi}{\varphi}
\usepackage[normalem]{ulem}
\usepackage{microtype}

\makeatletter
\newcommand\RedeclareMathOperator{
  \@ifstar{\def\rmo@s{m}\rmo@redeclare}{\def\rmo@s{o}\rmo@redeclare}
}
\newcommand\rmo@redeclare[2]{
  \begingroup \escapechar\m@ne\xdef\@gtempa{{\string#1}}\endgroup
  \expandafter\@ifundefined\@gtempa
     {\@latex@error{\noexpand#1undefined}\@ehc}
     \relax
  \expandafter\rmo@declmathop\rmo@s{#1}{#2}}
\newcommand\rmo@declmathop[3]{
  \DeclareRobustCommand{#2}{\qopname\newmcodes@#1{#3}}
}
\@onlypreamble\RedeclareMathOperator
\makeatother
\newcommand{\set}[1]{\{#1\}}                    
\newcommand{\setof}[2]{\{{#1}\mid{#2}\}}        
\newcommand{\bigsetof}[2]{\left\{{#1}\ \biggr|\ {#2}\right\}}        

\DeclareMathOperator{\Span}{{span}}

\RedeclareMathOperator{\dim}{{dim}}

\newcommand{\B}{\mathbb B} 
\newcommand{\Q}{\mathbb Q} 
\newcommand{\N}{\mathbb N} 
\newcommand{\Z}{\mathbb Z} 
 
\renewcommand{\S}{\mathbb S}

\newcommand{\calI}{\mathcal I}

\newcommand{\calR}{\mathcal R}

\newcommand{\eat}[1]{}
\renewcommand{\epsilon}{\varepsilon}

\newcommand{\var}{\mathit{var}}

\newcommand{\supp}{\mathit{supp}} 
\newcommand{\arity}{\mathit{ar}}

\newcommand{\tw}{\textit{tw}\xspace}
\newcommand{\ctw}{\textit{ctw}\xspace}
\newcommand{\atw}{\textit{atw}\xspace}
\newcommand{\ftw}{\textit{fhw}\xspace}
\newcommand{\fctw}{\textit{fcw}\xspace}

\newcommand{\0}{\mathbf{0}}
\newcommand{\1}{\mathbf{1}}

\newcommand{\adom}{\text{adom}}
\newcommand{\field}{\mathbb{F}}
\newcommand{\sring}{\mathbb{S}}
\newcommand{\inner}[2]{\langle #1 \mid #2 \rangle}
\newcommand{\sem}[1]{\llbracket{}{#1}\rrbracket{}}

\newcommand{\ceil}[1]{\lceil#1\rceil}
\usetikzlibrary{fit}
\usetikzlibrary{calc}
\usetikzlibrary{trees}

\newcommand{\activespace}{V^{\text{act}}}

%% file: 0-abstract.tex
We study query evaluation for an extension of sum-product queries (SPQ) that allows atoms with linear combinations of variables (e.g. $A(2X+3Y-Z, Y+Z)$), which we call convolution sum-product queries (CSPQs).
These queries arise in both practical settings
(e.g. image processing workloads) and theoretical ones (e.g. the
$(\min,+)$-convolution and $k$-SUM conjectures), and 
capture SPQs with added linear equality constraints. 
While prior work has
considered the impact of linear \textit{in-} and \textit{dis-}equality constraints on
query evaluation, the techniques developed in that setting are
asymptotically sub-optimal for CSPQs. 
To address this, we describe
several evaluation algorithms for CSPQs that leverage linear algebra techniques like rank analysis,
quotient spaces, and variable substitution. First, we adapt Worst-Case
Optimal Joins to CSPQs, achieving a runtime similar in spirit to that for
conjunctive queries. Then, we extend the definition of tree
decompositions (TDs) to CSPQs, and describe a factorized execution.  
In this extension, linear combinations are first-class citizens and play the same role as variables in traditional TDs.
We
define three width measures that bound the complexity of this
execution with respect to the size of the domain values, of the active
domain, and of the relation's support. Lastly, we show how these
methods can be applied to arbitrary fields beyond the rationals.

%% file: 1-introduction.tex
\section{Introduction}
\label{sec:intro}
The literature on query evaluation has extended conjunctive queries (CQs) in a variety of ways to capture new problems of interest. Sum-product queries (SPQs) generalize CQs from the Boolean semiring to an arbitrary semiring~\cite{DBLP:conf/pods/GreenKT07}, allowing them to express operations such as matrix multiplication, while CQs with disequalities ($\not =$) enabled the analysis of constraint satisfaction and subgraph-isomorphism problems within the CQ paradigm.
In this work, we investigate an extension of conjunctive queries with linear expressions over variables in the query.  For example, the convolution of two arrays $A, B$ can be expressed as $C(X) = \sum_Y A(Y)\cdot B(X-Y)$.  We call this class of queries \emph{Convolution Sum-Product Queries} (CSPQs).  We are motivated by a large class of applications, in a variety of domains, which already use linear expressions.  For example, SQL  allows arbitrary expressions over attributes in the WHERE clause.  Query optimizers based on dynamic programming already incorporate joins involving linear expressions, like $R_1.A+R_2.B = R_4.D+R_5.E-R_3.C$, see~\cite{DBLP:conf/sigmod/MoerkotteN08}.  Linear expressions on the variables also occur in a large class of important workloads in tensor, image, and signal processing. For example, matrix convolution is a key operation in image processing pipelines and is not expressible via standard sum-product queries. In this operation, a smaller pattern matrix is iteratively applied to patches of the larger image matrix $A$ in the following manner:
\begin{align*}
    B(i,j)= &\sum_{l,k}A(i+l,j+k)\cdot \text{Pattern}(l,k)
\end{align*}
By selecting an appropriate pattern matrix, this operation is the core computational primitive in image processing tasks such as edge detection and blurring, scientific simulations that use stencil computations, and in the convolution layers of neural networks~\cite{DBLP:conf/sc/DattaMVWCOPSY08,getreuer2009contour, DBLP:conf/ics/HolewinskiPS12, DBLP:conf/nips/KrizhevskySH12, DBLP:conf/pldi/Ragan-KelleyBAPDA13, DBLP:journals/toms/LuporiniLLKWHYK20}.
Several languages or formalisms already consider linear expressions, but typically as second-class citizens, represented as predicates separate from the main body of the query. 
Some fine-grained complexity conjectures are expressable in this way~\cite{DBLP:journals/talg/CyganMWW19,DBLP:journals/corr/abs-2404-04369}.  For example, the 3SUM conjecture states that, given three sets $A,B,C$, of size $n$, checking if there exists $i \in A,j \in B,k \in C$ so that $i+j=k$ cannot be done in time $O(n^{2-\varepsilon})$ and can be expressed as:
\begin{align*}
  Q() = & \bigvee_{i,j,k} A(i)\wedge B(j) \wedge C(k) \; \text{ such that } i+j=k
\end{align*}
while the MinConv conjecture states that the following cannot be evaluated in time $O(n^{2-\varepsilon})$:
\begin{align*}
  Q(k) = &  \min_{i,j} (A(i)+ B(j)) \; \text{ such that } i+j=k
\end{align*}
In this paper, we introduce Convolution Sum-Product Queries, and describe new algorithms for evaluating them.  CSPQs extend Sum-Product Queries with linear expressions, while keeping them first-class citizens. The linear expressions occur directly in the query's atoms, for example, the 3SUM and the MinCov queries above are written as follows in our formalism:
\begin{align}
  Q() \leftarrow {} & A(X) \wedge B(Y) \wedge C(X+Y) & Q(K) \leftarrow{} & A(I)+ B(K-I) \label{eq:3sum}
\end{align}
We define the semantics of CSPQs, showing that they have always a well-defined answer.  This is not obvious a priori.  For example, the query $Q() = R(X+Y)$ would normally be interpreted as $\sem{Q}= \sum_{x,y} R(x+y)$, which results in an infinite sum.  We use the concept of \emph{quotient space} from linear algebra to ensure that every CSPQs has a well-defined semantics.  In Sec.~\ref{sec:cspq} we introduce CSPQs formally, define their semantics, and discuss some of their properties.
Next, we describe several algorithms for evaluating a CSPQ in Sec.~\ref{sec:basicAlgorithm}.  We first describe a simple algorithm that runs in time $O(|\adom(D)|^k)$, where $\adom(D)$ is the active domain of the (support of the) input database, and $k$ is the dimension of the query's \emph{span}, (defined in Sec.~\ref{sec:cspq}). This algorithm is best suited for \emph{dense} databases, as found, for example, in processing dense tensors.  Second, we describe an algorithm, which is best suited for sparse databases, which runs in time $O(|D|^{\rho^*(Q)})$.  Here $\rho^*(Q)$ is the fractional edge covering number of a CSPQ, and is defined as a significant generalization of the traditional fractional edge covering number of a conjunctive query, or a hypergraph.  This algorithm is quite non-obvious.  For example, consider the following query, which is a slight variation of the 3SUM query above (Eq.~\eqref{eq:3sum}):
\begin{align*}
  Q() \leftarrow {}& A(X,Y) \land B(Y,Z) \land C(X+Y, Y+Z)
\end{align*}
We show that $\rho^*(Q)=3/2$, and therefore $Q$ can be evaluated in time $O(|D|^{3/2})$, were $|D|$ is the size of the database (number of tuples in $A,B,C$).  Although our definition of $\rho^*$ (in Def.~\ref{def:fractional:edge:cover}) looks very different from the traditional one, we show that it coincides with the traditional one when the CSPQ is a traditional conjunctive query.
Finally, we adapt tree decompositions to CSPQs.  Similarly to $\rho^*$, the definition of a tree decomposition requires a careful extension to linear spaces: we present the definition of a \emph{Convolution Tree Decomposition} (CTD) in Sec.~\ref{sec:decomposition}, and describe factorized evaluation algorithms based on CTDs in Sec.~\ref{sec:algorithm}. We also extend our techniques to arbitrary fields in Sec.~\ref{sec:otherFields}, then conclude.
{\bf Related Work.} Other extensions of conjunctive queries with numerical operations or predicates can be found in the literature.  In all cases the arithmetic operations are considered second-class citizens, designed to be included in predicates, separate from the main body of the conjunctive query.  For example, Arithmetic Constrains are simple comparisons $X<Y$ over a dense domain; Klug~\cite{DBLP:journals/jacm/Klug88}, Afrati et al.~\cite{DBLP:conf/edbt/AfratiLM04,DBLP:journals/tcs/AfratiLM06}, and Brisaboa et al.~\cite{DBLP:conf/adbis/BrisaboaHPP98} study the containment and the query rewriting using views problems under this extension.  Kiani and Shiri~\cite{DBLP:conf/otm/KianiS05,DBLP:conf/issads/KianiS05,DBLP:conf/ihis/KianiS05} study the query answering using views problem for conjunctive queries extended with linear equality constraints (which are equivalent to our CSPQs).  Srivastava~\cite{DBLP:journals/amai/Srivastava93}, and Ibarra and Su~\cite{DBLP:journals/jcss/IbarraS99} study the containment (subsumption) problem for queries with linear constraints, like $2X+3Y < 7$.
The query evaluation problem has been considered in several papers.  Extensions of conjunctive queries with general constraints has been discussed in~\cite{DBLP:conf/pods/KhamisCM0NOS19,DBLP:journals/tods/KhamisCMNNOS20}, and in~\cite{DBLP:journals/sigmod/0001023,DBLP:journals/tods/WangY26}.  These references apply a technique due to Willard to handle complex predicates, such as $X^2+XY > 5Z$, but since they target more general queries than CSPQs, their runtimes are worse than what we propose in this paper: see the discussion in Sec.~\ref{sec:basicAlgorithm}.  Koutris et al.~\cite{DBLP:journals/mst/KoutrisMRS17} and Abo Khamis et al.~\cite{DBLP:conf/icdt/KhamisNOS19} introduce specialized algorithms for evaluating conjunctive queries with disequality predicates, $X \neq Y$.  In contrast to these lines of work, the algorithms we present in this paper treat the linear expressions as first class citizens, and show how to apply worst-case optimal joins~\cite{DBLP:journals/sigmod/NgoRR13,DBLP:conf/icdt/Veldhuizen14,DBLP:conf/pods/NgoPRR12} and tree decompositions~\cite{DBLP:journals/jcss/GottlobLS02}, by combining them with core ideas in linear algebra.
Constraint databases, which have been heavily studied in the literature~\cite{DBLP:books/sp/Kuper00}, are unrelated to our paper.  They consider input relations that are infinite, defined implicitly using linear or algebraic constraints.  In our paper we always assume the support of the input relations is finite.
Beyond databases, expressions equivalent to our CSPQs over dense tensors have been intensively study by the compilers community. This work has resulted in the polyhedral model, which can represent various transformations of loops over arrays with linear index expressions, and has been leveraged in several widely used high-performance compilers for tensor algebra~\cite{DBLP:journals/taco/TavarageriHAKGU21,DBLP:conf/cgo/BaghdadiRRSAZSK19,DBLP:conf/osdi/ChenMJZYSCWHCGK18,DBLP:journals/taco/VerdoolaegeJCGTC13}. Our work applies similar linear transformations but goes further by analyzing them also in the sparse setting and extending them to leverage tree decompositions style factorizations.
\textbf{Summary of Results.}
\vspace{-.2em}
\begin{itemize}
    \item We extend sum-product queries to allow linear expressions on the variables, keeping them as first-class citizens (see Sec.~\ref{sec:cspq}).
    \item We extend the notion of tree decompositions to CSPQs and prove that these induce a factorization of the queries.
    (see Sec.~\ref{sec:decomposition})
    \item We analyze the complexity of CSPQs.
    We provide basic algorithms (essentially for full queries) and modify worst-case optimal join algorithms to achieve AGM-style bounds (see Sec.~\ref{sec:basicAlgorithm}). 
    Then, we apply these algorithms to the factorizations induced by our novel extension to tree decomposition.
    The complexity of the evaluation algorithm is characterized by three new width parameters that capture different input density regimes (see Sec.~\ref{sec:algorithm}).
\end{itemize}
\ifArxiv
Most proof are deferred to the appendix and in  Appendix~\ref{app:le} we explain precisely how CSPQs capture SPQs with added linear equality constraints.
\fi

%% file: 2-background.tex
\section{Background}
\label{sec:background}
In this section we introduce the notation used throughout the paper. 
If $\bm a=(a_1,\dots,a_n)$ is a tuple, then we write $|\bm a|:=n$, and denote by $f(\bm a)$ the tuple $(f(a_1), \ldots, f(a_n))$, where $f$ is any function.  
With some abuse, we will blur the distinction between a tuple $\bm a$ and the set of its values.  For example, if $\bm a=(a_1,\ldots,a_n)$ and $\bm b=(b_1,\ldots,b_m)$, then $\bm a \subseteq \bm b$ means $\set{a_1,\ldots,a_n} \subseteq \set{b_1,\ldots,b_m}$.
{\bf Databases and Sum-Product Queries.} We follow the definition of K-relations from~\cite{DBLP:conf/pods/GreenKT07}, where a relation is a function from tuples to a commutative semiring.  Throughout this paper, the domain of the tuples are the rational numbers $\Q$ (unless otherwise stated), and the commutative semiring is denoted by $\S=(\S,\oplus, \otimes, \bm 0, \bm 1)$.  Thus, if $R$ is a relation symbol of arity $r:=\arity(R)$, then an \emph{$\S$-relation instance}, or simply $\S$-relation, is a function $R^D\colon \Q^{r}\rightarrow\S$ with finite support, where the support is $\supp(R^D):=\{\bm x\in \Q^{r}\mid R^D(\bm x)\neq \0\}$.  We say that tuples $\bm x\in \supp(R^D)$ are \textit{annotated} by $R(\bm x)$.  We assume that $R^D$ is represented as a list of tuple-annotation pairs, of size $|R^D| := |\supp(R^D)|$.  
The active domain $\adom(R)\subseteq \Q$ of $R$ is the set of all domain elements that appear in the support.
A \emph{database instance} $D$ over a vocabulary $R_1, \ldots, R_m$ is tuple of $\S$-relations, $D = (R_1^D, \dots, R_m^D)$.  Its active domain is $\adom(D)=\bigcup_i \adom(R_i^D)$ and its size is $|D|:=\sum_i|R_i^D|$.
If $R$ is a relation symbol of arity $r$, and $\bm X$ is an $r$-tuple of variables, then the expression $R(\bm X)$ is called an \emph{atom}.  A \emph{sum-product query (SPQ)} is an expression of the form:
    \begin{align}
      Q(\bm X_0) \leftarrow R_1(\bm X_1)\otimes\ldots\otimes R_m(\bm X_m) \label{eq:spq}
    \end{align}
    where $R_1(\bm X_1),\ldots,R_m(\bm X_m)$ are atoms, and $\bm X_0 \subseteq \bigcup_{i=1,m} \bm X_i$ (called the \emph{safety condition}).  We write $\var(Q):=\bigcup_{i=1,m} \bm X_i$ for set of all variables in $Q$, and call $\bm X_0$ the \emph{head variables} of $Q$.  The \emph{arity of $Q$} is $|\bm X_0|$.  A \emph{conjunctive query} (CQ) is an SPQ interpreted in the Boolean semiring $\B$: we write $\wedge$ for $\otimes$.
    Fix a database instance $D$ and query $Q$.  A \emph{valuation} is a function $v : \var(Q) \rightarrow \Q$ that maps each variable to a value in the domain; $v$ maps each atom $R_i(\bm X_i)$ to a semiring value, namely to $R^D_i(v(\bm X_i))$.  The answer of $Q$ over $D$ is the $\S$-relation $\sem{Q}^D$ obtained by summing over all valuations:
\begin{align}
\forall \bm x\in &\Q^{|\bm X_0|}:&    \sem{Q}^D(\bm x) := & \bigoplus_{v\colon \var(Q)\rightarrow \Q , \text{ s.t. } v(\bm X_0)=\bm x}\left(\bigotimes_{i=1,m} R^D_i(v(\bm X_i))\right).
    \label{eq:spqsem}
\end{align}
We sometimes omit the superscripts $D$ and write simply $\sem{Q}$.  Note that $\sem{Q}$ has finite support.
{\bf Linear Algebra.} We fix a $d$-dimensional vector space $V$ over the rationals $\Q$, and write $\dim(V):=d$ for its dimension.  For any set $A\subseteq V$ we write $\Span(A)$ for the smallest vector subspace containing $A$. Given two subspaces $V_1,V_2\subseteq V$, $V_1+V_2$ denotes $\Span (V_1 \cup V_2)$.  Further, for a subspace $P\subseteq V$, $V/P$ is the quotient space of $V$ modulo $P$, which consists of equivalence classes where two vectors $v,v'\in V$ are in the same class when $v-v'\in P$.  Its dimension is $\dim(V/P)=\dim(V)-\dim(P)$.  We write $[v]_P$ for the equivalence class of $v$, or just $[v]$ when $P$ is clear from the context.
When $V=\Q^d$, then we denote by  $\inner{u}{v}$ the dot product of vectors  $u,v \in V$.  Given two sets of vectors $A, B \subseteq V$, we write $A \perp B$ when $\inner{x}{y}=0$ for all $x \in A, y \in B$, and write $A^\bot := \setof{v\in V}{\forall x\in A: x \perp v}$ for the subspace orthogonal to $A$.  Its dimension is $\dim(A^\bot) = d - \dim(\Span A)$. 
If $v \in V$ is fixed, then $\inner{-}{v}$ is a linear function $V \rightarrow \Q$, defined by $u \mapsto \inner{u}{v}$.  We extend this function to a tuple of vectors, $\inner{-}{v}: V^n \rightarrow \Q^n$: if $\bm A =(u_1,\dots,u_n) \in V^n$ then $\inner{\bm A}{v} := (\inner{u_1}{v}, \dots, \inner{u_n}{v}) \in \Q^n$.  If $B \subseteq V$ then we define $\inner{-}{[v]_{B^\bot}}: \Span B \rightarrow \Q$ as follows: for all $u \in \Span B$, $\inner{u}{[v]_{B^\bot}}:=\inner{u}{v}$.
{This is well defined in the following sense: if $[v]_{B^\bot}=[v']_{B^\bot}$ then $v-v' \in B^\bot$; since $u \in \Span B$, it follows that $\inner{u}{v-v'}=0$, or equivalently $\inner{u}{v}=\inner{u}{v'}$, proving that the definition of $\inner{u}{[v]_{B^\bot}}$ does not depend on the representative from $[v]_{B^\bot}$.}
{\bf Computational Model.}
We assume the RAM model where each cell holds a single domain or semiring element. 
Every operation in classical or semiring arithmetic is assumed to take constant time. 
We omit the single log factor associated with looking up tuples from relations or sorting them and write upper bounds as $O$ rather than $\tilde O$. 
We focus solely on the data complexity of queries in this work, so we assume queries to be fixed.

%% file: 3-cspq.tex
\section{Convolution Sum-Product Queries}
\label{sec:cspq}
In this section, we define \emph{convolution sum-product queries} (CSPQs).
The goal is to extend SPQs to allow atoms with linear expressions, like in Sec.~\ref{sec:intro}.  
For example, consider the atom $R(X+2Y,2X+Z)$.  
Using dot products we can write $X+2Y=\inner{(1,2,0)}{(X,Y,Z)}$ and $2X+Z=\inner{(2,0,1)}{(X,Y,Z)}$, and the atom becomes $R(\inner{\bm A}{(X,Y,Z)})$, where $\bm A$ is the tuple of vectors $(1,2,0)$ and $(2,0,1)$. 
For conciseness, we omit the variables and abbreviate the atom by $R(\bm A)$. 
In general, if $R$ is any relation symbol of arity $r = \arity(R)$, then we define a \emph{convolution atom} to be an expression $R(\bm A)$, where $\bm A$ is a tuple of $r$ vectors in $\Q^d$ (or, equivalently, a $d \times r$ matrix).
{Here, $d$ is the total number of variables.}
{
A convolution sum-product query is defined as follows (we first give the full definition and then delve into the intuition):
}
\begin{definition} \label{def:cspq}
  A \emph{convolution sum-product query} (CSPQ) over $\mathbb{Q}^d$ is an expression of the form:
    \begin{align}
    Q(\bm A_0)\leftarrow {} &  R_1(\bm A_1) \otimes \dots \otimes R_m(\bm A_m) \label{eq:cspq}
    \end{align}
    where each $R_i(\bm A_i)$, $i=1,m$ is a convolution atom, and the following \emph{safety condition} holds:\footnote{Recall that we identify a tuple $\bm A_0 = (v_1, \ldots, v_k)$ with a set $\set{v_1,\ldots, v_k}$: the safety condition states $v_j \in \Span(Q)$, $\forall j$.} $\bm A_0 \subseteq \Span(Q)$, where the \emph{span of $Q$} is $\Span(Q) := \Span(\bigcup_{i=1,m} \bm A_i)$.  The \emph{arity} of $Q$ is $\arity(Q) := |\bm A_0|$.
\end{definition}
{For a brief example, the query $Q() \leftarrow A((1,0)) \wedge B((0,1)) \wedge C((1,1))$ is the same as the 3SUM query in Eq.~\eqref{eq:3sum}; we will return to it shortly.}
When defining a CSPQ we may omit the head attributes $\bm A_0$ when they are implicit and the focus is on the body structure. 
Given a CSPQ over $V := \Q^d$, we associate the following subspaces: the \emph{null space} is $(\Span(Q))^\perp$, and the \emph{active space} is $\activespace := \Q^d/(\Span(Q))^\perp$.  Notice that $\dim(\Span(Q))=\dim(\activespace)$.  For a convolution atom $R(\bm A)$, we will write, with some abuse, $\Span(R)$ for $\Span(\bm A)$, and write $\Span(\calR)=\bigcup_{R \in \calR}(\Span(R))$ for a set of atoms $\calR$.
Note the similarity of an SPQ given in Eq.~\eqref{eq:spq} and a CSPQ given in Eq.~\eqref{eq:cspq}:
Vectors in CSPQs play the role of variables in SPQs.
This will be the running theme of this paper.
A CSPQ is interpreted in a semiring $\S$ as follows:
\begin{definition}\label{def:semantics}
  Let $Q$ be a CSPQ given by Eq.~\eqref{eq:cspq}, and let $\activespace$ be its active space. The answer of $Q$ over an $\S$-database instance $D=(R_1^D, \ldots, R_m^D)$ is the $\S$-relation $\sem{Q}^D$:
    \begin{align}
      \forall \bm x \in & \Q^{|\bm A_0|}:&     \sem{Q}^D(\bm x):=\bigoplus_{v\in \activespace\colon \inner{\bm A_0}{v} = \bm x}\bigotimes_{i=1,m}R^D_i(\inner{\bm A_i}{v}). \label{eq:semantic}
    \end{align}
    As usual, we omit the superscript $D$ when clear from the context, and write simply $\sem{Q}$.
\end{definition}
The notation $\inner{\bm A_i}{v}$ in Eq.~\eqref{eq:semantic} is well-defined for $v \in \activespace = V/(\Span(Q))^\perp$, because $\bm A_i \perp (\Span(Q))^\perp$. We prove in Appendix~\ref{app:sec3} that $\sem{Q}^D$ has finite support and that there are only finitely many non-$\0$ terms in the summation~\eqref{eq:semantic}:
\begin{restatable}{lemma}{lemsound}
\label{lem:sound}
$\sem{Q}^D$ is an $\S$-relation (meaning: it is well-defined and has finite support).
\end{restatable}
{
\begin{proof}[Proof Sketch.]
    Every non-$\0$ term on the right hand side of Equation~\eqref{eq:semantic} for some
    $\sem{Q}(\bm x)$ has all its factors non-$\0$, since $\0$ is absorbing.
    Hence each combination
    $(\inner{\bm A_1}{v},\ldots,\inner{\bm A_m}{v})$ occurring in a
    non-$\0$ term belongs to the finite product of the relation supports.
    Only finitely many such combinations are possible. If two corresponding
    elements $v,v'\in\activespace$ yield the same combination, then
    $v-v'$ is in $\Span(Q)^\perp=\Span(\bm A_1\cup\dots\cup\bm A_m)^\perp$, so $v=v'$ in $\activespace$.
    Thus the sum defining $\sem{Q}(\bm x)$ has finitely many non-$\0$ terms.
    Moreover, $\sem{Q}(\bm x)$ is non-$\0$ only if there is at least some non-$\0$ term.
    However, there are only finitely many $v\in \activespace$ that lead to a non-$\0$ term and $v$ determines $\bm x$ through $\inner{\bm A_0}{v}=\bm x$.
\end{proof}
}
As explained above, the vectors $\bm A_i$ in an atom $R_i(\bm A_i)$ indicate which variables are used in the atom, and with what multiplicities.  To make the analogy even clearer, let $X_1,\dots,X_d$ be $d$ variables, then an atom $R_i(\bm A_i)$ represents the atom $R_i(\inner{\bm A_i}{(X_1,\ldots,X_d)})$.  The expression $R_i(\inner{\bm A_i}{v})$ in Eq.~\eqref{eq:semantic} with $v=[(v_1,\dots,v_d)]$ should be interpreted as applying the valuation $\{X_j\mapsto v_j\mid j=1,d\}$ to the variables, as in the semantics of traditional SPQs given in Eq.~\eqref{eq:spqsem}.
When $\Span(Q)=\Q^d$, then $(\Span(Q))^\perp=\set{0}$ and the active space $\activespace = \Q^d/\set{0}$ is isomorphic to $\Q^d$.  Then, the summation in Eq.~\eqref{eq:semantic} can be replaced with a summation over vectors $v \in \Q^d$.  We illustrate this with a simple example:
\begin{example} \label{ex:3sum} The 3SUM query from Sec.~\ref{sec:intro} Eq.~\eqref{eq:3sum} is a query interpreted in the Boolean semiring, and is written as a CSPQ over $\Q^2$ as follows
\begin{align*}
  Q() \leftarrow & A((1,0)) \wedge B((0,1)) \wedge C((1,1)).
\end{align*}
As discussed above, we can write $Q$ more naturally as:
  $Q() \leftarrow  A(X) \wedge B(Y) \wedge C(X+Y).$
We will often use this notation for readability.  Here the {span of $Q$ is $\Q^2=\Span\{(1,0),(0,1),(1,1)\}$}, and the semantics in Eq.~\eqref{eq:semantic} can be written  as $\sem{Q}=\bigvee_{x,y\in\Q}\left(A(x)\wedge B(y) \wedge C(x+y)\right)$.
\end{example}
In this paper, we study the query evaluation problem for CSPQs.  The ``intuitive'' way to compute the answer $\sem{Q}$ in Eq.~\eqref{eq:semantic} is as follows.  Maintain a hash table $T\colon \Q^{|\bm A_0|}\rightarrow \sring$ (initialized with $\supp(T)=\emptyset$), iterate over the active space $v\in \activespace$, compute $s_v :=\bigotimes_i R_i(\inner{\bm A_i}{v})$, then increment $T(\inner{\bm A_0}{v}) := T(\inner{\bm A_0}{v}) \oplus s_v$.  Once all vectors $v$ have been processed, return $\sem{Q}:=T$.  Of course, this not a real algorithm, because the summation over the active space $\activespace$ is an infinite sum: we describe practical algorithms in Sections~\ref{sec:basicAlgorithm} and~\ref{sec:algorithm}.  Here, we illustrate the basic idea with an example.
\begin{figure}
    \centering
    \begin{minipage}[t]{.11\textwidth}
    \centering
    \caption*{$A$}  
    \vspace{.3em}
    \begin{tabular}{ |c|c}
        \cline{1-1}
         $1$ & 1  \\
         $2$ & 3  \\
         $3$ & 4  \\
         $4$ & 5  \\
         \cline{1-1}
    \end{tabular}
    \end{minipage}
    \begin{minipage}[t]{.18\textwidth}
    \centering
    \caption*{$B$}  
    \vspace{.3em}
    \begin{tabular}{|cc|c}
        \cline{1-2}
         -1 & 0 & 1  \\
         0 & -1 & 5  \\
         0 & 1 & 3  \\
         \cline{1-2}
    \end{tabular}
    \end{minipage}
    \begin{minipage}[t]{.11\textwidth}
    \centering
    \caption*{$C$}  
    \vspace{.3em}
    \begin{tabular}{|c|c}
        \cline{1-1}
         $1$ & 1  \\
         $2$ & 3  \\
         $3$ & 4  \\
         $4$ & 5  \\
         \cline{1-1}
    \end{tabular}
    \end{minipage} 
    \begin{minipage}[t]{.1\textwidth}
    \end{minipage}
    \begin{minipage}[t]{.49\textwidth}
    \centering
    \caption*{$Q(2X+2Z)\leftarrow A(X)\cdot B(Y-X,Z-Y)\cdot C(Z)$}  
    \vspace{.3em}
    \begin{tabular}{|c|c}
        \cline{1-1}
        $6$ & 27\\
        $10$ & 108\\
        $14$ & 160\\
        \cline{1-1}
    \end{tabular}
    \end{minipage}    
    \caption{Example Evaluation of a Convolution Sum-Product Query.}
    \label{fig:ex:spqle}
\end{figure}
\begin{example} 
\label{ex:output}
Consider the query\footnote{If we used the Boolean semiring, then the query could be written equivalently using arithmetic predicates as follows: $Q(W)\leftarrow A(X)\wedge B(Y_1,Y_2) \wedge C(Z) \wedge (X+Y_1+Y_2=Z)\wedge (W=2X+2Z)$.} $Q$,  given in Fig.~\ref{fig:ex:spqle}, over $\Q^3$ and interpreted in the semiring of natural numbers $(\N, +, \cdot, 0, 1)$.
The $\N$-relation instances are bags of rational numbers, $A,C\colon \mathbb{Q}\rightarrow \N,$ and $ B\colon \mathbb{Q}^2\rightarrow \N$
and shown 
in Fig.~\ref{fig:ex:spqle}.  To evaluate $Q$, we iterate over triples $x,y,z\in \mathbb{Q}$ such that $A(x)\neq 0, B(y-x,z-y)\neq 0, $ and $C(z)\neq 0$.  
To determine these triples, we can iterate over all values $a \in \supp(A)$, $(b_1,b_2) \in \supp(B)$, $c \in \supp(C)$, and solve the overdetermined system of linear equations $\left(\begin{smallmatrix}1 & 0 & 0\\-1 & 1 & 0\\0 & -1 & 1 \\ 0 & 0 & 1\end{smallmatrix}\right) \left(\begin{smallmatrix}x\\y\\z\end{smallmatrix}\right)=\left(\begin{smallmatrix}a \\ b_1 \\ b_2 \\ c\end{smallmatrix}\right)$.  The $4\cdot 3\cdot 4$ choices for the RHS 
lead to 9 solutions:
$$(x,y,z)= \quad (1,1,2),(2,1,1), (2,2,1), \quad (2,2,3), (3,2,2), (3,3,2),  \quad (3,3,4), (4,3,3), (4,4,3)$$
We grouped the solutions by the value of $2x+2z$; the output $\sem{Q}(2x+2z)$ is the sum of their annotations.  For example, in the first group $2x+2z=6$, and $\sem{Q}(6)=A(1)B(0,1)C(2)+A(2)B(-1,0)C(1)+A(2)B(0,-1)C(1)=9+3+15=27$.  The complete answer is in Fig.~\ref{fig:ex:spqle}.
\end{example}
{\bf The Active Space.} Fix a CSPQ query $Q$ over $\Q^d$, and denote by $S := \Span(Q)$.  In our examples so far, $S = \Q^d$. In this case $S^\perp=\set{0}$ and the active space is $\activespace = \Q^d/S^\perp \simeq \Q^d$.  We could have defined the query's semantics in Eq.~\eqref{eq:semantic} as a sum over the entire space $\Q^d$ instead of a sum over $\activespace$.  But when $S \subsetneq \Q^d$, then summing over the entire space leads to an infinite result.  Any two vectors $v,v'\in \Q^d$ that satisfy $[v]_{S^\perp} = [v']_{S^\perp}$ produce the same result, because $\inner{\bm A_i}{v}=\inner{\bm A_i}{v'}$ for all $i=1,m$.  Since each equivalence class $[v]_{S^\perp}$ is infinite, the summation in Eq.~\eqref{eq:semantic} would add the same term infinitely many times.  This is why we introduced the active space. Summing over the active space guarantees that each combination of tuples from $R_1, \ldots, R_m$ appears only once in the sum.  We illustrate this with a simple example.
\begin{example} \label{ex:single:atom} Consider the query $Q() \leftarrow R(X+Y)$ over $\Q^2$, which just sums up all annotations of the tuples in $R$.  Written formally, the atom is $R((1,1))$, and $S := \Span(Q) =\Span(\set{(1,1)}) = \setof{(z,z)}{z \in \Q}$ has dimension 1.  Obviously, we cannot define $\sem{Q}=\sum_{x,y \in \Q}R(x+y)$, since this will add every tuple in $R$ infinitely many times.  The insight here is that two vectors $v=(x,y)$ and $v'=(x',y')$ represent the same tuple in $R$ iff $x+y = x'+y'$ or, equivalently, iff $\inner{(1,1)}{v}=\inner{(1,1)}{v'}$, which we can also write as $(1,1) \perp (v-v')$, or $[v]_{S^\perp}=[v']_{S^\perp}$. By summing over equivalence classes $[v]_{S^\perp}$ instead of over vectors $v$, Eq.~\eqref{eq:semantic} includes every tuple of $R$ only once.  Since $[(x,y)]_{S^\perp}=[(x+y,0)]_{S^\perp}$, every equivalence class has a canonical representative $(z,0)$, and we can also express $\sem{Q}$ as $\bigoplus_{[(z,0)] \in \activespace} R(z+0)=\bigoplus_{z \in \Q}R(z)$.  This captures our intuition that{, because $\Span(Q)$ is 1-dimensional,} the expression $X+Y$ in the query should be substituted by a single variable $Z$ and written as $Q()\leftarrow R(Z)$.
\end{example}
We allow $\Span(Q) \subsetneq \Q^d$ in order to evaluate a query via a sequence of subqueries.  For example, consider the standard 2-path Boolean query $Q()\leftarrow R(X,Y)\wedge S(Y,Z)$, which, viewed a CSPQ is over the 3-dimensional space $\Q^3$.  We can compute $Q$ in linear time using the following factorization:
$Q_1()\leftarrow {} R(X,Y)\wedge Q_2(Y),\; Q_2(Y)\leftarrow {} S(Y,Z).$
While $Q$ is defined over $\Q^3$, the two subqueries have a span of dimension 2 and $\sem{Q}=\sem{Q_1}$.  We expand on this idea in Section~\ref{sec:decomposition}.
{
{\bf The Safety Condition.}
Intuitively, the requirement $\bm A_0 \subseteq \Span(Q)$ in the definition of a CSQP (Def.~\ref{def:cspq}) corresponds to the traditional safety condition for conjunctive queries: every variable in the head must occur in the body.  For example, consider the unsafe query: $Q(X,Y,Z) \leftarrow R(X,Y) \wedge S(Y)$.  
In vector notation the query is $Q(e_1,e_2,e_3) \leftarrow R(e_1,e_2)\wedge S(e_2)$, where $e_1=(1,0,0), e_2=(0,1,0), e_3=(0,0,1)$.   Here $\set{e_1,e_2,e_3} \not\subseteq \Span(e_1,e_2)$, reflecting the fact that the query is unsafe.
However, the requirement is also technically important because, for elements $[v]\in \activespace$, the expression $\inner{\bm A_0}{[v]}$ is only well-defined when $\bm A_0\subseteq \Span(Q)$.
}
{\bf Convolution-free CSPQ.} We say that a CSPQ $Q(\bm A_0)\leftarrow \bigotimes_{i=1,m}R_i(\bm A_i)$ over $\Q^d$ is \emph{convolution-free} if $\bigcup_{i=1,m} \bm A_i = \set{e_1, \ldots, e_d}$, where $e_1=(1,0,\ldots,0), \ldots, e_d = (0,0,\ldots,1)$ are the canonical basis vectors in $\Q^d$.  (There is no restriction on $\bm A_0$.)  When we replace the $e_j$ with variables $X_j$, then $Q$ is simply a traditional SPQ; examples include matrix multiplication $Q(X,Z)\leftarrow A(X,Y)\otimes B(Y,Z)$, or any traditional conjunctive query, e.g. $Q(Y)\leftarrow R(X)\wedge S(X,Y)$.  Thus, CSPQs generalize SPQs.
{We illustrate with an example.
\begin{example}
  \label{ex:categorical} Consider the following CSPQ over $\Q^3$, interpreted in the natural numbers $\mathbb{N}$:
  \begin{align}
    Q() \leftarrow & R(X,Y) \cdot S(Y,Z) \cdot T(X,Z) \label{eq:triangle}
  \end{align}
  Written in vector notation, the first atom becomes $R((1,0,0),(0,1,0))$, and similarly for the other two atoms, thus, it is convolution-free.  $\Span(Q)=\Q^3$, and $\activespace=\Q^3/\set{0}\simeq \Q^3$.  The traditional semantics of SPQs
  equals the CSPQ semantics in Eq.~\eqref{eq:semantic}: $\sem{Q} = \sum_{x,y,z \in \Q} R(x,y)\cdot S(y,z)\cdot T(x,z)$.
\end{example}
}
{\bf Invariance Under Substitutions.} Variable substitution is a key operation in linear algebra that facilitates the analysis of an expression by converting it to a simpler, equivalent form. This will be a crucial tool throughout this paper, and we prove here that this operation is valid on CSPQs. 
Formally, a \emph{substitution} for a CSPQ $Q$ is an injective linear function $\varphi: \Span(Q) \rightarrow V'$ for some vector space $V'$.
We write $Q^\varphi$ for the result of applying $\varphi$ to $Q$ as in the following lemma:
\begin{restatable}[Variable Substitution]{lemma} {lemmasubstitution}\label{lemma:substitution} 
Let $Q_1(\bm A_0)\leftarrow \bigotimes_{i=1,m} R_i(\bm A_i)$ be a CSQP over $V_1 := \Q^{d_1}$, and let $\varphi: \Span(Q_1) \rightarrow V_2$ be a substitution, where $V_2 := \Q^{d_2}$.  Define $\bm B_i := \varphi(\bm A_i)$ for $i=0,m$, and consider the query $Q_2(\bm B_0)\leftarrow \bigotimes_{i=1,m}R_i(\bm B_i)$.  Then $Q_1, Q_2$ are equivalent: $\forall D$, $\sem{Q_1}^D=\sem{Q_2}^D$.
\end{restatable}
\begin{proof}[Proof Sketch.]
  Denote $S_1 := \Span(Q_1)$, $S_2 := \Span(Q_2)$ and notice that $\phi\colon S_1\rightarrow S_2$ is a bijection.  
  Let $\activespace_1 :=V_1/S_1^\perp, \activespace_2:=V_2/S_2^\perp$ be the active spaces of $Q_1, Q_2$ respectively.  
  Then, there exists a bijection $\varphi^* : \activespace_2 \rightarrow \activespace_1$ s.t.\footnote{When both null spaces are $S_1^\perp=\set{0}$, $S_2^\perp=\set{0}$ then $\varphi^*$ is simply the \emph{adjoint} of $\varphi$ (the transposed matrix).}:
    $\inner{\varphi(a)}{v_2} =  \inner{a}{\varphi^*(v_2)}, \;\forall a \in S_1, \; \forall v_2 \in \activespace_2$. 
  For any $\bm x \in \Q^{|\bm A_0|}$
    \begin{align*}
      \sem{Q_1}(\bm x)&=  \bigoplus_{v_1\in \activespace_1\colon \inner{\bm A_0}{v_1}=\bm x}\bigotimes_{i=1,m}R^D_i(\inner{\bm A_i}{v_1}) = \bigoplus_{v_2\in \activespace_2\colon \inner{\bm A_0}{\varphi^*(v_2)}=\bm x}\bigotimes_{i=1,m}R^D_i(\inner{\bm A_i}{\varphi^*(v_2)})  
  \end{align*}
    For $i=0,m$, we have $\inner{\bm A_i}{\phi^*(v_2)}=\inner{\bm B_i}{v_2}$.
    Thus, $\sem{Q_1}(\bm x)=\sem{Q_2}(\bm x)$
    and $\sem{Q_1}=\sem{Q_2}$.
\end{proof}
{
\begin{example} \label{eq:output:continued}
  To showcase an application of Lemma~\ref{lemma:substitution}, let us reconsider the CSPQ $Q$ from Ex.~\ref{ex:output}: $Q(2X+2Z)\leftarrow A(X) \cdot B(Y-X, Z-Y) \cdot C(Z)$.  
  Perform the substitution $\varphi$ given by: $U=X, \, V =Y-X, \, W=Z-Y$, and rewrite $Q$ to $Q^\varphi(4U+2V+2W)\leftarrow A(U)\cdot B(V,W)\cdot C(U+V+W)$.
  Lemma~\ref{lemma:substitution} states that for any instances of $A,B,C$, the answers to $Q$ are the same as to $Q^\varphi$, i.e., $\sem{Q}=\sem{Q^\varphi}$.
  To compute $Q^\varphi$, it naturally suffices to iterate over values $u,v,w$, where $u$ is in the support of $A$ and $(v,w)$ is in the support of $B$.
  We only need to keep tuples where $u+v+w$ is in the support of $C$.
  For example, referring to the instances $A, B, C$ in Fig.~\ref{fig:ex:spqle}, it suffices to consider the following values to evaluate $Q^\varphi$:
    \begin{align*}
        (u,v,w)= \;\; (1,0,1),(2,-1,0), (2,0,-1),\;\; (2,0,1),(3,-1,0),(3,0,-1),\;\; (3,0,1),(4,-1,0),(4,0,-1)
    \end{align*}
    We grouped the tuples by the value of $4u+2v+2w$.  For each group we compute one output $\sem{Q^\varphi}(4u+2v+2w)$, as before.  The result
    is the same as $\sem{Q}$ depicted in Fig.~\ref{fig:ex:spqle}.
\end{example}
}

%% file: 4-basic-algorithm.tex
\section{Basic Algorithms}
\label{sec:basicAlgorithm}
In this section, we describe two basic algorithms for evaluating a CSPQs. Our first result captures evaluation of CSPQs over dense tensors:
\begin{theorem}[Dense Evaluation]
\label{thm:dense}
Let $Q$ be a CSPQ and $k=\dim(\Span(Q))$.  Then, $\sem{Q}^D$ can be computed in time $O(|\adom(D)|^k)$.
\end{theorem}
Our second result adapts a worst-case optimal join (WCOJ) algorithm~\cite{DBLP:journals/sigmod/NgoRR13,DBLP:conf/icdt/Veldhuizen14,DBLP:conf/pods/NgoPRR12} from CQs to CSPQs.  Recall that, if $Q \leftarrow \bigwedge_{i=1,m}R_i(\bm X_i)$ is a CQ and $\bm Y \subseteq \var(Q)$, then a \emph{fractional edge cover} of $\bm Y$ is a tuple of weights $\bm w=(w_1, \ldots, w_m)$ such that $\forall X_j\in \bm Y$, $\sum_{i: X_j \in \bm X_i}w_i \geq 1$.  The \emph{fractional edge cover number} is $\rho^*(\bm Y) = \min_{\bm w \in EC(\bm Y)} \sum_{i=1,m} w_i$, where $EC(\bm Y)$ is the set of fractional edge covers of $\bm Y$. WCOJ algorithms compute $\sem{Q}^D$ in time $O(|D|^{\rho^*(Q)})$, where $\rho^*(Q) := \rho^*(\var(Q))$.
We begin by extending the definition of a fractional edge cover from CQs to CSPQs.
\begin{definition} \label{def:fractional:edge:cover} Let $Q \leftarrow \bigotimes_{i=1,m}R_i(\bm A_i)$ be a CSPQ.
A fractional edge cover of a tuple of vectors $\bm C \subseteq \Span(Q)$ is a tuple of weights $\bm w=(w_1,\dots, w_m)$ where $w_i\in [0,1]$, satisfying the following:
  \begin{align}
    \bm C \subseteq & \Span \left(\bigcup_{I \in \calI} \left(\bigcap_{i\in I} \Span(\bm A_i)\right)\right) & \text{where } \calI := \bigsetof{I\subseteq \{1,\dots, m\}}{\sum_{i\in I}w_i\geq 1}. \label{eq:fractional:edge:cover}
  \end{align}
  We denote the \textit{fractional edge cover number} with $\rho^*_Q(\bm C):=\min_{\bm w \in EC(\bm C)} \sum_i w_i$.
  When the query is clear from the context then we write simply $\rho^*(\bm C)$, and denote $\rho^*:=\rho^*(Q) := \rho^*(\Span(Q))$.
\end{definition}
{The classical fractional edge cover of a conjunctive query $Q$ assigns weights to atoms so that every variable in $\var(Q)$ is covered by total weight at least one. 
We explain why, when restricted to convolution-free queries, Definition~\ref{def:fractional:edge:cover} coincides with the classical one.  Indeed, observe that a set of indices $I$ belongs to the set $\calI$ in Eq.\eqref{eq:fractional:edge:cover} iff the set of atoms $\setof{R_i}{i \in I}$ have a total weight $\geq 1$.  Since all vectors in $A_1, \ldots, A_m$ are canonical basis vectors, the intersection $\bigcap_{i\in I}\Span(\bm A_i)$ is the same as $\Span(\bigcap_{i \in I}\bm A_i)$.  Condition $\Span(Q) \subseteq \Span \left(\bigcup_{I \in \calI} \left(\bigcap_{i\in I} \Span(\bm A_i)\right)\right)$ says ``every canonical base vector $e_j$ (i.e. variable $X_j$) must be in some intersection $\bigcap_{i \in I}\bm A_i$, i.e. must belong to a set of atoms $R_i$ for which $\sum_i w_i \geq 1$''.  This is the same as saying that the variable $X_j$ is covered by $\bm w$.  For a general CSPQ,
for a set of atoms $I$, the intersection $\bigcap_{i\in I}\Span(\bm A_i)$ consists of the linear combinations that all atoms in $I$ can represent. The definition declares such a subspace available whenever the total weight of $I$ is at least one, and asks these subspaces to span the vectors being covered. Thus, $\rho^*$ is the minimum total weight needed to cover the vectors in $\Span(Q)$. When the vectors in the tuples $\bm A_i$ are canonical basis vectors, these intersections are precisely the variables shared by the atoms in $I$, and the definition reduces to the usual fractional edge cover.}
In Lemma~\ref{lemma:rhostar:cspq:cq}, we show how this generalizes the $\rho^*$ of a CQ.  Our second result in this section is:
\begin{restatable}[Sparse Evaluation]{theorem}{thmsparse}
  \label{thm:sparse} Fix a CSPQ $Q$. $\sem{Q}^D$ can be computed in time $O(|D|^{\rho^*(Q)})$.
\end{restatable}
In the rest of the section, we prove Theorems~\ref{thm:dense} and~\ref{thm:sparse}.  The proofs rely on a reduction from the evaluation of a CSPQ to a traditional CQ.
We start by considering a special case of a CSQP.
\begin{definition} \label{def:independent}
    A CSPQ $Q \leftarrow \bigotimes_{i=1,m} R_i(\bm A_i)$ is \emph{independent} if $\{b_1,\dots,b_k\}=\bigcup_{i=1,m} \bm A_i$ are linearly independent vectors. (Vectors $b_j$ are allowed to appear multiple times and we ignore the head.)
\end{definition}
Any convolution-free CSPQ (see Sec.~\ref{sec:cspq}) is an independent query.  Conversely, if $Q$ is independent, then we can obtain an equivalent convolution-free query by applying the substitution $\phi$ that maps $b_1, \ldots, b_k$ to the canonical basis $e_1, \ldots, e_k$. This query, $Q^\phi$, when viewed as a SPQ with variables $X_1,\dots,X_k$, can be computed using a naive algorithm with runtime $O(|\adom(D)^k|)$, or a WCOJ with runtime $O(|D|^{\rho^*(Q^\varphi)})$. Because $\sem{Q}=\sem{Q^\varphi}$, we can evaluate $Q$ in this runtime as well.
In general, in order to evaluate a CSPQ we \emph{guard} it with \textit{reducts}.  
To that end, we call an $\S$-relation $P^D : \Q^p \rightarrow \S$ \emph{Boolean} if $P^D(\bm x)=\bm 0$ or $P^D(\bm x)=\bm 1$ for all $\bm x \in \Q^p$; with some abuse, we will also view $P^D$ as being interpreted in the Boolean semiring $\mathbb{B}$.  Fix an atom $R(\bm A)$.  A \emph{reduct} of $R(\bm A)$ is an atom $P(\bm C)$ where $\bm C$ is a tuple of linearly independent vectors s.t. $\bm C \subseteq \Span(\bm A)$.  
Given an $\S$-relation instance $R^D$, we define a reduct of $R^D$ onto $\bm C$ as a Boolean $\S$-relation instance satisfying: $R^D(\inner{\bm A}{v})\neq \bm 0 \Rightarrow P^{D}(\inner{\bm C}{v})=\bm 1$, which implies $R^D(\inner{\bm A}{v})= R^D(\inner{\bm A}{v})\otimes P^{D}(\inner{\bm C}{v})$, for all $v \in \Q^d$: in other words, adding the reduct $P(\bm C)$ does not change the meaning of the atom $R(\bm A)$.
A reduct $P^D$ can be computed in time $O(|R^D|)$ as follows.  
Represent each vector $c_i \in \bm C$ as a linear combination
$c_i = \sum_j \lambda_{ij}a_j$ of the vectors in $\bm A= (a_1, \ldots, a_r)$.  Then, $P^D$ can be computed by iterating over each tuple $(u_1, \ldots, u_r) \in \supp(R^D)$ and inserting $(\sum_j \lambda_{1j}u_j, \ldots, \sum_j \lambda_{pj}u_j)$ in $\supp(P^D)$; notice that $|P^D| \leq |R^D|$.  
For example, if the atom is $R(X+Y,Y,Z+W)$ and the reduct is $P(X,X-Y)$ then $\supp(P^D) = \setof{(u_1-u_2, u_1-2u_2)}{(u_1,u_2,u_3)\in \supp(R^D)}$.
Intuitively, one can think of reducts as being a generalization of projections.
Given a CSPQ $Q$, we construct an independent (Def.~\ref{def:independent}) query $G$ whose atoms are reducts of atoms from $Q$.  We call $G$ a \emph{reduct} of $Q$ and interpret it in the Boolean semiring, indicated by using $\wedge$ instead of $\otimes$.  However, we see $\sem{G}$ as a Boolean $\sring$-relation, i.e. with $\bm 0, \bm 1$ values.
\begin{definition} \label{def:guard}
  A \emph{reduct} of $Q \leftarrow \bigotimes_{i=1,m}R_i(\bm A_i)$ is an independent CSPQ $G(\bm C_0) \leftarrow \bigwedge_{j=1,n}P_j(\bm C_j)$, where each $P_j(\bm C_j)$ is a reduct of some atom $R_i(\bm A_i)$.
\end{definition}
Observe that $\Span(\bm C_0) \subseteq \Span(G)\subseteq \Span(Q)$. Adding $G$ to $Q$ does not affect its semantics:
\begin{restatable}{lemma}{lemmaguardone}\label{lemma:guard:1} 
Let $G(\bm C_0)\leftarrow \bigwedge_{j=1,n}P_j(\bm C_j)$ be a reduct of $Q(\bm A_0)\leftarrow \bigotimes_{i=1,m} R_i(\bm A_i)$.  Then $Q$ is equivalent to $Q'(\bm A_0)\leftarrow \bigotimes_{i=1,m} R_i(\bm A_i)\otimes G(\bm C_0)$ in the following sense: for any instance $D$, if we instantiate the atom $G(\bm C_0)$ in $Q'$ as the answer of $G$ on $D$, $G^D := \sem{G}^D$, then $\sem{Q}^D=\sem{Q'}^{D}$.
\end{restatable}
The proof is immediate from the fact that adding a reduct does not change the meaning of the atom it is derived from.  When $\Span(\bm C_0)=\Span (Q)$, then we can compute $\sem{Q}=\sem{Q'}$ by simply iterating over the tuples in $\sem{G}$.  In such a case, we call $G$ a \textit{guard}:
\begin{definition}
  Fix a CSPQ $Q$.  An atom $G(\bm C_0)$ of $Q$ is a \textit{guard} if $\Span(Q)=\Span(\bm C_0)$.
\end{definition}
\begin{restatable}{lemma}{lemmaguardtwo}\label{lemma:guard:2} 
 If $Q$ has a guard $G(\bm C_0)$, then $\sem{Q}^D$ can be evaluated in time $O(|G^D|)$.
\end{restatable}
\begin{proof}[Proof Sketch.]
  In the definition of $\sem{Q}^D$ given by Eq.~\eqref{eq:semantic}, we can restrict the summation to vectors $v \in \activespace = \Q^d/(\Span(Q))^\perp$ that satisfy $\sem{G}(\inner{\bm C_0}{v})\neq \bm 0$.  There are at most $|G^D|$ such vectors, because $\inner{\bm C_0}{-}\colon \activespace \rightarrow \Q^{|\bm C_0|}$ is injective, as $(\Span(Q))^\perp=(\Span (\bm C_0))^\perp$.
\end{proof}
Lemmas~\ref{lemma:guard:1} and~\ref{lemma:guard:2} lead to an algorithm to evaluate $Q$: Construct a reduct $G$ of $Q$ that acts as a guard. Evaluate the reduct traditionally, then leverage it to evaluate $Q$. We now prove Theorem~\ref{thm:dense}.
\begin{proof}[Proof of Theorem~\ref{thm:dense}.]  Choose a basis $b_1, \ldots, b_k$ of $\Span(Q)$ such that each $b_j$ belongs to some tuple $\bm A_i$: this is possible since $\Span(\bigcup_i \bm A_i)=\Span(Q)$.  For each $b_j \in \bm A_i$, let $P_j(b_j)$ be a reduct of the atom $R_i(\bm A_i)$.  Observe that this reduct has size $|P_j^D| \leq |\adom(R_i^D)|$.  Form the reduct query $G(b_1, \ldots, b_k) \leftarrow \bigwedge_{j=1,k}P_j(b_j)$.  By applying the substitution $b_i \mapsto e_i$, we transform it into an equivalent (by Lemma~\ref{lemma:substitution}) traditional CQ $G^\phi$, which is a cross product that can be computed in time $O(\adom(D)^k)$.  The theorem then follows from Lemmas~\ref{lemma:guard:1} and~\ref{lemma:guard:2}.
\end{proof}
{
  \begin{example}\label{example:dense}
  We reconsider the CSPQ $Q$ from Ex.~\ref{ex:output}: $Q(2X+2Z)\leftarrow A(X) \cdot B(Y-X, Z-Y) \cdot C(Z)$.  
  A reduct of $Q$ is the CSPQ $G(X,Y-X,Z-Y)\leftarrow A(X) \cdot B(Y-X, Z-Y)$ and, due to Lemma~\ref{lemma:guard:1}, $Q$ is equivalent to $Q'(2X+2Z)\leftarrow A(X) \cdot B(Y-X, Z-Y) \cdot C(Z)\cdot G(X,Y-X,Z-Y)$ where $G$ is a guard of $Q'$.
  Since $G$ is independent, we can use a variable substitution to transform it into a convolution-free CSPQ.
  Namely, we can use the substitution: $U=X, \, V =Y-X, \, W=Z-Y$. 
  Thus, $G$ can be evaluated in time $O(|\adom(D)|^3)$.
  Using Lemma~\ref{lemma:guard:2}, we can also evaluate $Q'$ in time $O(|\adom(D)|^3)$ and $ \sem{Q}=\sem{Q'}$.
\end{example}
}
Next, we prove Theorem~\ref{thm:sparse}.  In order to evaluate $\sem{Q}^D$, we construct, as before, a reduct $G(\bm C_0)$ that acts as guard for $Q$ and use Lemmas~\ref{lemma:guard:1} and~\ref{lemma:guard:2}.  This time we need to relate the fractional edge covering number $\rho^*$ of the CSPQ $Q$ to the traditional $\rho^*$ of the CQ representing the reduct.  Thus, the proof of the theorem relies on the careful analysis of $\rho^*$.  First, we show substitution invariance:
\begin{lemma} \label{lemma:same:rhostar} If $Q$ is a CSPQ, $\bm C \subseteq \Span(Q)$, and $\varphi$ is a substitution, then $\rho^*_Q(\bm C) = \rho^*_{Q^\varphi}(\varphi(\bm C))$.  
\end{lemma}
\begin{proof}
  The lemma follows from Def.~\ref{def:fractional:edge:cover} and the fact that $\Span(\varphi(\bm A_i))=\varphi(\Span(\bm A_i))$.
\end{proof}
Next, we prove that $\rho^*$ for convolution-free CSPQs coincides with $\rho^*$ for CQs:
\begin{lemma} \label{lemma:rhostar:cspq:cq} Let $Q \leftarrow \bigotimes_{i=1,m}R_i(\bm A_i)$ be a convolution-free CSPQ,
and let $Q' \leftarrow \bigwedge_{i=1,m} R_i(\bm X_i)$ be the CQ obtained by replacing each $e_j$ with a variable $X_j$ (we ignore the head variables in both queries).  Then, for any subset of variables $X_{j_1}, \ldots, X_{j_k}$ the following holds: $\rho^*(e_{j_1}, \ldots, e_{j_k})=\rho^*(X_{j_1}, \ldots, X_{j_k})$.  
\end{lemma}
\begin{proof}
  For any $I\subseteq \{1,\dots, m\}$ the vector space $\bigcap_{i\in I}\Span(\bm A_i)$ is simply the vector space spanned by the canonical basis vectors $e_j$ for which $e_j\in \bm A_i$ for all $i\in I$.  Thus, for $\bm w$ to be a fractional edge cover of $\set{e_{j_1}, \ldots, e_{j_k}}$, it has to be the case that for every $e_j\in \set{e_{j_1}, \ldots, e_{j_k}}$ there exists a $I\subseteq \{1,\dots, m\}$ such that $\sum_{i\in I}w_i\geq 1$ and $e_j\in \bm A_i$ for all $i\in I$.  Therefore, $\bm w$ is also a traditional fractional edge cover of $\rho^*(X_{i_1},\dots ,X_{i_k})$.  Conversely, for a traditional fractional edge cover $\bm w$, it holds that $\sum_{i: e_j\in \bm A_i}w_i\geq 1 $ for any choice of $X_{j}\in\{X_{i_1},\dots ,X_{i_k}\}$.  Thus, $I:=\{i: e_j\in \bm A_i\}\in \mathcal{I}$ and $e_j\in \bigcap_{i\in I}\Span(\bm A_i)$.
\end{proof}
Thus, in order to evaluate $Q$ in time $O(|D|^{\rho^*})$ it suffices to construct a guard $G$ defined by a reduct $G(\bm C_0) \leftarrow \cdots$ with the same $\rho^*$ as $Q$, because, after substitution, the reduct is essentially a traditional CQ with the same $\rho^*$.  Next, we show how to construct such a reduct.
\begin{lemma} \label{lemma:rhostar:geq} Fix a CSPQ $Q \leftarrow \bigotimes_{i=1,m} R_i(\bm A_i)$.  For any tuple of vectors $\bm C \subseteq \Span(Q)$ there exists a reduct $G(\bm C) \leftarrow \bigwedge_{i=1,m} P_i(\bm C_i)$ of $Q$ such that $\rho^*_{G}(\bm C)=\rho^*_Q(\bm C)$.
\end{lemma}
\begin{proof}
  Let $\bm w$ be an optimal fractional edge cover of $\bm C$ w.r.t. $Q$, and denote by $\calI= \setof{I \subseteq [m]}{\sum_{i \in I}w_i \geq 1}$.  Let $b_1, \ldots, b_{k}$ be a basis for the subspace $\Span\left(\bigcup_{I \in \calI}\left(\bigcap_{i \in I} \Span(\bm A_i)\right)\right)$, such that, for each $b_j$, there exists $I_j \in \calI$, $b_j \in \bigcap_{i \in I_j}\Span(\bm A_i)$.
  For each atom $R_i(\bm A_i)$, define $\bm C_i$ the tuple of basis vectors $b_j$ that are in $\Span(\bm A_i)$, i.e. $\bm C_i = \Span(\bm A_i) \cap \set{b_1, \ldots, b_{k}}$, let $P_i(\bm C_i)$ be a reduct of the atom $R_i(\bm A_i)$, and define $G\leftarrow \bigwedge_{i=1,m} P_i(\bm C_i)$. To prove $\rho^*_{G}(\bm C) \leq \rho^*_Q(\bm C)$, we show that $\bm w$ is a valid fractional edge cover for $G$. 
  We check that $\bm C \subseteq \Span(\bigcup_{I \in \calI} \bigcap_{i \in I} \Span(\bm C_i))$ by noting that every basis vector $b_j$ is in the set $\bigcup_{I \in \calI} \bigcap_{i \in I} \Span(\bm C_i)$: this holds because, by construction, $b_j \in \bigcap_{i \in I_j} \Span(\bm A_i)$, and, by the definition of $\bm C_i$, we also have $b_j \in \bigcap_{i \in I_j} \Span(\bm C_i)$.  This proves $\rho^*_{G}(\bm C) \leq \rho^*_Q(\bm C)$, and equality follows as $\Span(\bm C_i)\subseteq \Span(\bm A_i) $ for all $i$.
\end{proof}
We can now prove Theorem~\ref{thm:sparse}.
\begin{proof}[Proof of Theorem~\ref{thm:sparse}.] 
  Set $\bm C := \bigcup_{i=1,m}\bm A_i$, and use Lemma~\ref{lemma:rhostar:geq} to obtain a reduct $G$ s.t. $\rho^*(Q)=\rho^*(G)$ and $G(\bm C)$, viewed as an atom, is a guard for $Q$.  We apply a substitution $\varphi$ to obtain a convolution-free query $G^\varphi$ with $\rho^* := \rho^*(G^\varphi) = \rho^*(G)=\rho^*(Q)$.  We use a WCOJ algorithm to compute $\sem{G^\varphi}^D$ in time $O(|D|^{\rho^*})$, and observe that $\sem{G}^D=\sem{G^\varphi}^D$ (Lemma~\ref{lemma:substitution}).  Finally, we use Lemmas~\ref{lemma:guard:1} and~\ref{lemma:guard:2} to compute $\sem{Q}^D$ in the same time.
\end{proof}
\begin{example} Consider the following query over $\Q^3$:
    \begin{align}
    \label{ex:triangle}
    Q(X,Y,Z)\leftarrow A(X,Y)\otimes B(Y,Z) \otimes C(X+Y,Y+Z)        
    \end{align}
    Theorem~\ref{thm:dense} leads to an algorithm with runtime $O(|\adom(D)|^3)$.  The weights $\bm w=(w_A,w_B,w_C)=(1,1,0)$ form a fractional edge cover of $Q$ (all variables are in $A,B$), which leads to an algorithm with runtime $O(|D|^2)$.  But we can do even better, by observing that $\bm w=(\frac{1}{2}, \frac{1}{2}, \frac{1}{2})$ is also a fractional edge cover.\footnote{We verify cond.~\eqref{eq:fractional:edge:cover}: As $\{\{A,B\}, \{B,C\}, \{C,A\}\}\subseteq \mathcal{I}$, it suffices to check that $\left(\Span A \cap \Span B\right) \cup \left(\Span B \cap \Span C\right) \cup \left(\Span C \cap \Span A\right)$ spans $\Q^3$.  We have $\Span A = \setof{(u,v,0)}{u,v}$, $\Span B=\setof{(0,v,w)}{v,w}$, $\Span C=\setof{(u,u+w,w)}{u,w}$ and the intersections are $\setof{(0,v,0)}{v }$, $\setof{(0,w,w)}{w}$, $\setof{(u,u,0)}{u}$: together, they span all of $\Q^3$.}  Then, by Theorem~\ref{thm:sparse}, we can evaluate the query in time $O(|D|^{3/2})$.
    We sketch the algorithm.  
    We start by applying the substitution: $U=X+Y,\; V=Y,\; W=Y+Z$, and the query becomes $Q^\varphi(U-V,V,W-V) \leftarrow A(U-V,V)\otimes B(V,W-V)\otimes C(U,W)$.  
    At this point we take the reduct of each atom onto the variables that are functionally determined by its attributes and obtain:
    $G(U,V,W)\leftarrow P_A(U,V)\land P_B(V,W) \land P_C(U,W).$
    This query $G$ can be evaluated in time $O(|D|^{3/2})$ using Generic Join~\cite{DBLP:journals/sigmod/NgoRR13}, then we can iterate over the $(u,v,w)\in \sem{G}$ answer tuples to compute $\sem{Q^\varphi}=\sem{Q}$ in the same time.
\end{example}
\begin{paragraph}{Discussion.} 
  CSPQs are special cases of queries with general constraints discussed in FAQ-AI~\cite{DBLP:conf/pods/KhamisCM0NOS19,DBLP:journals/tods/KhamisCMNNOS20,DBLP:journals/sigmod/0001023,DBLP:journals/tods/WangY26}.  Query~\eqref{ex:triangle} above is expressed in that formalism by replacing $X+Y,Y+Z$ with new variables $W,V$, and adding four inequality conditions:
\begin{align*}
    Q \leftarrow {} & A(X,Y)\land B(Y,Z) \land C(W,V) \land  W\leq X+Y\land   W\geq X+Y\land  V\leq Y+Z\land  V\geq Y+Z
\end{align*}
Previous techniques evaluate this query by starting from a tree decomposition that ensures that all variables of any constraint are included in two adjacent bags.  In particular, two adjacent bags must contain $X,Y,W$, and similarly from $Y,Z,V$.  The tree width is 2, leading to a runtime of $O(|D|^2)$; as we saw, our algorithm runs in time $O(|D|^{3/2})$.
{At a high level, since CSPQs are a strict subset of the queries in FAQ-AI, they enable query evaluation techniques like Theorem~\ref{thm:sparse} that are not possible in the general setting of FAQ-AI.  For example, the conclusion section of~\cite{DBLP:journals/corr/KhamisNRR15} discusses the query $Q \leftarrow R(X,Y)\wedge S(Y,Z)\wedge X+Y=Z$, observing that their current framework requires $O(N^2)$ time to evaluate the query (where $N=|R|+|S|$).  Expressed as a CSQP the query becomes $Q \leftarrow R(X,Y) \wedge S(Y,X+Y)$, and its fractional edge cover is $\rho^*(Q)=1$, because $\bm w = (1,0)$ is a fractional edge cover of $Q$. The algorithm implicit in the proof of Theorem~\ref{thm:sparse} evaluates this query in time $O(N)$.}
\end{paragraph}

%% file: 5-tree-decompositions.tex
\section{Tree Decompositions for CSPQ}
\label{sec:decomposition}
Tree decompositions (TDs) speedup  the  evaluation of a traditional CQ by factorizing it and applying WCOJs to each factor \cite{DBLP:journals/talg/GroheM14, DBLP:conf/pods/NgoPRR12, DBLP:conf/focs/AtseriasGM08}. In this section, we adapt the notion of TDs to CSPQs.
In our new definition the bags are sets of vectors, and the connectedness condition (also called the running intersection property) incorporates the structure of the vector space.\footnote{We study query evaluation only under worst-case complexity, where the complexity depends only on the input size, not the output size.  For that reason, we require all output variables to be part of a bag.}
\begin{definition}
\label{def:CTD}
Let $Q(\bm A_0)\leftarrow \bigotimes_{i=1,m} R_i(\bm A_i)$ be a CSPQ.  A \textit{Convolution Tree Decomposition} (CTD) of $Q$ is a pair $(T,\chi)$ where $T$ is a tree, and $\chi$ is a function that maps each node $v\in V(T)$ to a finite set of vectors $\chi(v)\subseteq \Span(Q)$ (called the \textit{bag} of $v$) satisfying:
    \begin{enumerate}
    \item For every $i=0,m$ there exists a $v\in V(T)$ such that $\bm A_i \subseteq \Span \chi(v)$.  We choose one such $v$ and call it the \textit{cover} of $\bm A_i$; the \emph{root} of $T$ is the node $r$ that covers $\bm A_0$.
    \item For each edge $uv\in E(T)$, $\Span \chi(T_u) \cap \Span\chi(T_v) = \Span \chi (u) \cap \Span \chi(v)$, where $T_u, T_v$ are the subtrees obtained by removing the edge $uv$,  $\chi(T_u):=\bigcup_{u'\in V(T_u)}\chi(u')$, eq. for $\chi(T_v)$.
    \end{enumerate}
\end{definition}
{A CTD generalizes TDs by replacing variables with vectors. In an ordinary tree decomposition, an atom is covered by a bag if all its variables occur in that bag; here, its argument vectors $\bm A_i$ must lie in the bag's span. The running-intersection condition is generalized in the same way: every vector in the intersection of the spans of the two sides of an edge must lie in the spans of both adjacent bags. This intersection is the interface through which the factorized computation passes messages. Any traditional TD can be viewed as a CTD by replacing each variable with its corresponding canonical basis vector. Under this identification, the coverage and running-intersection conditions are precisely those of traditional TDs.}
\begin{example}
\label{ex:tripleConvolution}
 Consider the following query:
    $Q(X+Y)\leftarrow A(X)\otimes B(Y)\otimes C(Z-X-Y)\otimes D(Z).$
    Clearly, this query can be solved in cubic time (e.g. using Theorem~\ref{thm:dense}), and one can check that $\rho^*(Q)=3$, thus both algorithms in Sec.~\ref{sec:basicAlgorithm} require cubic time.  However, we can solve this query in quadratic time, using the following CTD $(T,\chi)$:
    \begin{center}
    \begin{tikzpicture}[scale = 0.9, font=\footnotesize]
        \node[rectangle, draw] (ABCDE) at (3,0.25) {
        $X,Y$
        };
        \node[rectangle, draw] (DEFGH) at (0,0) {
        $Z-X-Y, Z$
        };
        \draw (ABCDE) -- (DEFGH) node[midway, above] {$X+Y$};
    \end{tikzpicture}
    \end{center} 
    There are two nodes: the root $\chi(r)=\set{X,Y}$ covering $A,B$ and a child $\chi(c)=\{Z-X-Y,Z\}$ covering $C,D$.  We also show a basis for their intersection, $X+Y$.
    We can use this CTD to factorize $Q$ into two subqueries:
    $\;Q^c(X+Y)\leftarrow C(Z-X-Y)\otimes D(Z).$ and $Q^r(X+Y)\leftarrow A(X)\otimes B(Y) \otimes Q^c(X+Y)$ 
    To show the correctness of this factorization, that is $\sem{Q^r}=\sem{Q}$, we perform the variable substitution $U=X,\;V=X+Y,\;W=Z$ and distribute:
    \begin{align*}
        \sem{Q^c}(v) = \bigoplus_{w\in \Q} C(w-v)\otimes D(w), \;\; \sem{Q}(v)
        &= \bigoplus_{u,w\in \Q} A(u)\otimes B(v-u)\otimes C(w-v)\otimes D(w) = \sem{Q^r}(v)
    \end{align*}
    These two subqueries $Q^c,Q^r$ can be solved in quadratic time, so $Q$ can be as well.
\end{example}
In general, we use CTDs to factorize CSPQs, generalizing the way that traditional TDs are used to factorize SPQs.  The key technical tool is the following lemma, which we prove in the appendix:
\begin{restatable}[Factorization Lemma]{lemma}{lemmafactorization}\label{lemma:factorization}
  Let $Q(\bm A_0)\leftarrow \bigotimes_{i=1,m} R_i(\bm A_i)\otimes \bigotimes_{j=1,n} S_j(\bm B_j)$ be a CSPQ and consider the following two queries:
  $$Q_1(\bm A_0)\leftarrow \bigotimes_{i=1,m} R_i(\bm A_i) \otimes Q_2(\bm B_0), \quad Q_2(\bm B_0)\leftarrow \bigotimes_{j=1,n} S_j(\bm B_j)$$
  where $\bm B_0$ is a tuple of vectors s.t. $\Span \bm B_0 = \left(\Span \bigcup_{i=0,m} \bm A_i\right) \cap \left(\Span \bigcup_{j=1,n} \bm B_j\right)$.
  Then, on any database $D$, we have $\sem{Q}^D=\sem{Q_1}^D$, where the atom $Q_2(\bm B_0)$ in $Q_1$ is instantiated as $\sem{Q_2}^D$.
\end{restatable}
The lemma says that we can split the atoms of $Q$ arbitrarily into $R_1, R_2, \ldots$ and $S_1, S_2, \ldots$, evaluate first the subquery consisting of atoms $S_j$, then use that result together with the atoms $R_i$ to get the final answer.  The only condition is that the output vectors $\bm B_0$  must cover the intersection of the two subqueries. This is where we need to allow the span of a query to be a strict suspace of $\Q^d$: in general, both $\Span(Q_1)$ and $\Span(Q_2)$ may be strict subspace of $\Q^d$, even when $\Span(Q)=\Q^d$.
We will apply Lemma~\ref{lemma:factorization} to fully factorize $Q$ using a CTD $(T,\chi)$. Denote $r$ its root, choose any child $c$, and let $T_r, T_c$ be the subtrees after removing the edge $rc$.  Denote by $R_i(\bm A_i)$ the atoms covered by $T_r$, and by $S_j(\bm B_j)$ the atoms covered by $T_c$.  Then we can write the query as $Q(\bm A_0)\leftarrow \bigotimes_iR_i(\bm A_i)\otimes \bigotimes_jS_j(\bm B_j)$ and apply Lemma~\ref{lemma:factorization}, by choosing any basis $\bm B_0$ of $\Span \chi(c)\cap \Span \chi(r)$. We then recursively factorize the subqueries, $Q_1, Q_2$, using the fact that they have strictly smaller CTDs, namely $T_c$ and $T_r$ respectively.  When the CTD becomes a single node $v$, then we are done.
However, like for traditional CQ, this runs into slight technical complications when the bag $\chi(v)$ is larger than the span of the atoms covered by $v$.  This problem is exacerbated by CSPQs.  For example, a bag $\chi(v)=\set{X,Y}$ may cover an atom $R(X)$ but  no atom for the variable $Y$. The subquery at $v$ would use fewer variables than mentioned in $\chi(v)$, violating the safety condition in Def.~\ref{def:cspq}.  This is resolved by shrinking bags' spans to only the space required to 1) cover relations and 2) maintain connectedness.
\begin{definition}
\label{def:nice}
Let $(T,\chi)$ be a CTD of the CSPQ $Q(\bm A_0)\leftarrow \bigotimes_iR_i(\bm A_i)$.  For all $v\in V(T)$, define $\zeta(v):=\bigcup_{i\colon \bm A_i \text{ is covered by }v}\bm A_i$, and extend $\zeta$ to subtrees $T'$ of $T$ via $\zeta(T'):=\bigcup_{v\in V(T')}\zeta(v)$.  We say that $(T,\chi)$ is a \textit{minimal} CTD when it additionally satisfies the following:
    \begin{itemize}
    \item[(3)] For each $v\in V(T)$, we have (where $T_u,T_v$ represent the two subtrees of $T-uv$)
        \begin{align}
            \Span \chi(v) = \Span \zeta(v) + \sum_{u \text{ is adjacent to } v} (\Span \zeta(T_u)\cap \Span \zeta(T_v)). \label{eq:niceTD}
        \end{align}
    \end{itemize}
\end{definition}
In any CTD, we have $\zeta(v) \subseteq \chi(v)$, and, $\Span \zeta(T_u) \cap \Span \zeta(T_v) \subseteq \Span \chi(T_u) \cap \Span \chi(T_v) \subseteq \Span \chi(v) $ (by the connectedness condition).  In a minimal CTD, $\chi(v)$ is only as large as needed to satisfy these two conditions.  We note that any CTD can be converted into a minimal CTD:
\begin{restatable}{lemma}{lemnice}
\label{lem:nice}
    Let $(T,\chi)$ be a CTD of $Q$.
    Then, there exists a minimal CTD $(T,\chi')$ of $Q$ where $\Span \chi'(v) \subseteq \Span \chi(v)$ for all $v\in V(T)$.    
\end{restatable}
For minimal CTDs, we can now give the full factorization (the proof is in the appendix).
\begin{restatable}[Query Factorization]{theorem}{thmcorrectness}
\label{thm:correctness}
Let $Q(\bm A_0)\leftarrow \bigotimes_i R_i(\bm A_i)$ be a CSPQ and $(T,\chi)$ be a minimal CTD with root $r$.  For each node $v \in V(T)$ we define the following subquery:
    \begin{align}
        Q^v(\bm B^v)\leftarrow & \bigotimes_{v \text{ covers } R_i(\bm A_i)}R_i(\bm A_i)\otimes \bigotimes_{c\in child(v)} Q^c(\bm B^c). \label{eq:tdquery}
    \end{align}
    where $\bm B^r := \bm A_0$, and for $v \neq r, \bm B^v$ is any basis of $\Span \chi(v)\cap \Span \chi(\text{parent}(v))$.
    Then, $\sem{Q}=\sem{Q^r}$.
\end{restatable}
Because $(T,\chi)$ is a minimal CTD, we have $\Span \chi(v) = \Span (Q^v)$ for all $v\in V(T)$.
We call $\sem{Q^c}$ the \emph{message} passed from the child $c$ to the parent $v$.

%% file: 6-query-evaluation.tex
\section{Tree Decomposition Based Query Evaluation and Width Measures}
\label{sec:algorithm}
In this section we show how to use CTDs to evaluate a CSPQ.  While the basic idea is similar to evaluating traditional CQs using traditional TDs, CSPQ create a few new challenges. For example, the size of a message relation of arity $k$ may exceed $O(|adom(D)|^k)$.  We describe three algorithms that tackle these challenges with runtimes characterized by three different width measures.
\begin{example} \label{ex:longConvolution} Throughout this section we will illustrate with the following running example.  The following query generalizes the $(k+1)$-SUM query to an interpretation in an arbitrary semiring:
    \begin{align}
      Q()\leftarrow {}& A_1(X_1)\otimes\cdots \otimes A_k(X_k)\otimes B(X_1+X_2+\cdots+X_k) \label{eq:ksum}
    \end{align}
    where $|A_1|, \ldots, |A_k|, |B| \leq n$.  Recall that a fine-grained complexity assumption is that $Q$ cannot be computed in time polynomially lower than $O(n^{\ceil{(k+1)/2}})$ \cite{DBLP:conf/icalp/AbboudL13, DBLP:conf/icalp/LincolnWWW16}: none of our algorithms break this lower bound, but they generalize differently to arbitrary CSPQs.  Consider the following CTD:
    \begin{center}
    \begin{tikzpicture}[scale = 0.9, font=\footnotesize]
        \node[rectangle, draw] (1) at (0,0) {
        $X_1,X_2$
        };
        \node[rectangle, draw] (2) at (3,0.075) {
        $X_1+X_2,X_3$
        };
        \node[rectangle, draw] (3) at (7.5,0.1875) {
        $X_1+X_2+X_3,X_4$
        };
        \node[rectangle, draw] (4) at (13,0.325) {
        $X_1+\dots +X_{k-1},X_k$
        };
        \draw (1) -- (2) node[midway, above] {$X_1+X_2$};
        \draw (2) -- (3) node[midway, above] {$X_1+X_2+X_3$};
        \draw[dotted] (3) -- (4);
    \end{tikzpicture}
    \end{center} 
    There are $k-1$ nodes, $v_1, v_2, \ldots, v_{k-1}$, which cover $\set{A_1,A_2}, \set{A_3}, \ldots, \set{A_{k-1}}, \set{A_k,B}$ respectively.  Then, Theorem~\ref{thm:correctness} leads to the following bottom-up (or left-to-right in our drawing) factorization:
    \begin{align*}
        &Q^{v_1}(X_1+X_2)\leftarrow A_1(X_1)\otimes A_2(X_2),\quad Q^{v_i}(X_1+\dots +X_{i+1})\leftarrow A_{i+1}(X_{i+1}) \otimes Q^{v_{i-1}}(X_1+\dots +X_i),\\
        &Q^{v_{k-1}}()\leftarrow A_k(X_k)\otimes B(X_1+\dots + X_k) \otimes  Q^{v_{k-2}}(X_1+\dots + X_{k-1})
    \end{align*}
    We can compute $\sem{Q^{v_1}}$ in time $O(|A_1|\cdot |A_2|)=O(n^2)$, but $\sem{Q^{v_1}}$ can have size $n^2$.  The message sizes increase in $i$, $|\sem{Q^{v_i}}| \leq n^{i+1}$, leading to a complexity of $O(n^k)$ for the naive bottom-up algorithm.  
  \end{example}
\subsection{Evaluation Algorithm with respect to the Largest Integer}
\label{sec:algorithm:integers}
We start by examining the effect of the values in the active domain. In this section, we assume the active domain consists of integers; this is without generality, because we can multiply all rational numbers in the active domain with the least common multiple of their denominators, compute $Q$, then divide all domain values of the result by the same least common multiple.  When $\adom(D) \subseteq \Z$, then we define $n(D) := \max_{z \in \adom(D)} |z|$.  We show here that $Q$ can be computed in time $O(n(D)^{\ctw(Q)+1})$, where $\ctw(Q)$ is the convolution tree width defined as follows:
\begin{definition}[Convolution Treewidth]\label{def:ctw}
  Let $(T,\chi)$ be a CTD of the CSPQ $Q$. The \textit{convolution treewidth} of $(T,\chi)$, and of $Q$, are defined as
    $$\ctw(T,\chi)=\max_{v\in V(T)}\dim \Span \chi(v)-1, \quad \ctw(Q)=\min_{(T,\chi)}\ctw(T,\chi).$$
\end{definition}
{Convolution treewidth measures a bag by the dimension of its span. The convolution treewidth is the maximum such dimension minus one; as with treewidth, this subtraction makes the runtime exponent $\ctw(Q)+1$. For integer-valued databases, intuitively, a $k$-dimensional bag has $O(n(D)^k)$ possible values.}
\begin{restatable}{theorem}{thmalgorithmIntegers}
\label{thm:algorithmIntegers}
Let $Q$ be a CSPQ and let $D$ be a database with integer domain values, i.e., $\adom(D)\subseteq \mathbb{Z}$.  Then $\sem{Q}^D$ can be computed in time $O(n(D)^{\ctw(Q)+1})$.
\end{restatable}
\begin{proof}[Proof Sketch]
  Let $(T,\chi)$ be a width optimal CTD of $Q$.  We evaluate $Q$ bottom-up, using the Query Factorization Theorem~\ref{thm:correctness}, where each message $Q^v$ is computed using the dense algorithms given by Theorem~\ref{thm:dense}.  Consider the query $Q^v(\bm B^v)$ shown in Eq.~\eqref{eq:tdquery} and notice that (assume $v\neq r$ here) we are allowed to pick the basis $\bm B^v$ as we please. 
  We pick the $b\in \bm B^v$ to be integral linear combinations of the vectors in the body of $Q^v$.
  Then, if $n$ is the largest input domain value occurring in any $\supp(R_i^D)$ or $\supp(\sem{Q^c}^D)$, then the output values in $\supp(\sem{Q^v}^D)$ are integers and of size $O(n)$, because each such value is an integral linear combination of integers (the picked integral linear combination is query dependent and, thus, only affects the constant hidden by the big $O$ notation).
By Lemma~\ref{lem:nice} we can assume that $(T,\chi)$ is minimal, which implies $\dim(Q^v)= \dim \Span \chi(v)$, therefore, by induction along $T$, query~\eqref{eq:tdquery} can be computed in time $O(n^{\ctw(Q)+1})$ by Theorem~\ref{thm:dense}.
\end{proof}
\begin{example} \label{ex:longConvolution:1} Continuing Example~\ref{ex:longConvolution}, we notice that the convolution treewidth is $1$.  Assume that $\adom(A_i) \subseteq \Z$, let $n_i$ be the largest value of $\adom(A_i)$, and set $n := \sum_i n_i$.  Then the largest value in $\adom(\sem{Q^{v_i}})$ is $n_1+\ldots+n_i$, thus $Q$ can be computed in time $O(n^2) = O(n(D)^2)$.\footnote{Of course, this does not contradict the $(k+1)$-SUM conjecture, because this runtime depends on the maximum value in the database instead of on the number of values $|\adom(D)|$ as in Example~\ref{ex:longConvolution}.  In particular, if instead we choose the input relation $A_i$ to have active domain $\adom(A_i) = \setof{j\cdot n^{i-1}}{j=0,1,\ldots,n-1}$, then $\supp(\sem{Q^{v_i}})=\{0,\dots, n^{i+1}-1\}$ has size $n^{i+1}$, leading to a total runtime of $O(n^k)=O(|\adom(D)|^k)$.}
\end{example}
\subsection{Evaluation Algorithm with respect to the Active Domain Size}
\label{sec:algorithm:activeDomain}
Next, we examine an alternative method to compute each bag of a CTD in a runtime that depends only on the size of active domain of the input, and not on those of the increased active domain of the messages.  
The exponent is given as follows:
\begin{definition}\label{def:vector-cover}
  Fix a CSPQ $Q(\bm A_0) \leftarrow \bigotimes_{i=1,m} R_i(\bm A_i)$ and
  let $\bm C_0\subseteq \Span (Q)$.  A \emph{vector cover} of $\bm C_0$ is a subset of vectors $\bm C\subseteq \bigcup_i \bm A_i$ such that $\bm C_0\subseteq \Span (\bm C)$.  We denote with $\alpha(\bm C_0):=\min_{\bm C} |{\bm C}|$ the minimal size of a vector cover of $\bm C$.  Observe that $\dim \Span (\bm C_0) \leq \alpha(\bm C_0)$.
\end{definition}
{A vector cover spans a set of vectors using body vectors. Its size can exceed $\dim \Span(\bm C_0)$: the body vectors that lie in $\Span(\bm C_0)$ may not suffice to span $\bm C_0$. A cover may therefore need vectors outside $\Span(\bm C_0)$, so $\Span(\bm C)$ can be strictly larger than $\Span(\bm C_0)$. This quantity captures how many body vectors are needed to cover the vectors in a bag.}
\begin{definition}[Active Treewidth]\label{def:atw}
  Let $(T,\chi)$ be a CTD of the CSPQ $Q$.  Then, the \textit{active treewidth} of $(T,\chi)$ and $Q$ are defined as
$$\atw(T,\chi)=\max_{v\in V(T)}\alpha(\chi(v))-1, \quad \atw(Q)=\min_{(T,\chi)}\atw(T,\chi).$$
\end{definition}
{Active treewidth measures a bag by the minimum number of body vectors needed to span it. As with treewidth, we subtract one from the maximum bag size. Unlike convolution treewidth, it reflects the cost of constructing a guard from the input relations rather than enumerating values in the bag's span. A vector cover of size $k$ yields a reduct guard computable in time $O(|\adom(D)|^k)$, which explains the exponent $\atw(Q)+1$.}
\begin{restatable}{theorem}{thmadom}
\label{thm:adom}
Let $Q$ be a CSPQ.
    Then, $\sem{Q}$ can be computed in time $O(|\adom(D)|^{\atw(Q)+1})$.
\end{restatable}
\begin{proof}[Proof Sketch.]
  Fix a CTD $(T,\chi)$ of $Q$.  We combine the results of Sec.~\ref{sec:basicAlgorithm} and Sec.~\ref{sec:decomposition} to obtain this result.  The idea is to first add reducts to $Q$ such that after applying the Factorization Theorem~\ref{thm:correctness} that all subqueries $Q^v$ are guarded.  Fix a node $v \in V(T)$, denote $\bm C^v_0 := \chi(v)$, and let $\bm C=(b_1,\dots,b_k)$ be a vector cover of $\bm C^v_0$ of size $k:=\alpha(\bm C^v_0)$.  For each $b_j$, let $R_i(\bm A_i)$ be an atom such that $b_j\in \bm A_i$.  We compute the reducts $P_j(b_j)$ of $R_i(\bm A_i)$, and denote with $G^v(\bm C^v_0)\leftarrow \bigwedge_{j=1,k} P_j(b_j)$ a reduct of $Q$.  By Theorem~\ref{thm:dense}, $\sem{G^v}^D$ can be computed in time $O(|\adom(D)|^{\alpha(\chi(v))})$.  By Lemma~\ref{lemma:guard:1}, we can add an atom $G^v(\bm C^v_0)$ to $Q$ without affecting the semantics.  We repeat this process for all tree nodes $v\in V(T)$.  Then, $(T,\chi)$ is still valid CTD of the modified query.
  We compute the query bottom-up on the tree, using the Factorization Theorem~\ref{thm:correctness}.  Each query $Q^v$ can be computed in time $O(|\sem{G^v}|)$ by Lemma~\ref{lemma:guard:2}, because it is guarded by $G^v$.  Thus, the runtime is dominated by the time needed to compute reducts, which is $O(|\adom(D)|^{\atw(Q)+1})$.
\end{proof}
\begin{example}
  Continuing Example~\ref{ex:longConvolution} consider the bag $\chi(v_i) = \set{X_1+\cdots+X_i, X_{i+1}}$. There are two relevant vector covers for $\chi(v_i)$: either $X_1, X_2, \ldots, X_{i+1}$, or $X_{i+1}, X_{i+2}, \ldots, X_k, X_1+\cdots+X_k$.  We choose the smallest cover for each and add reducts to the query~\eqref{eq:tdquery}.  
  The reducts are:
 \begin{alignat*}{2}
   &G^{v_1}(X_1,X_2) \leftarrow A_1\land A_2,&  \quad &G^{v_2} \leftarrow  A_1\land A_2\land A_3, \quad \dots\\
    \quad &G^{v_{k-2}}(X_{k-1},X_1+\dots +X_{k-2}) \leftarrow A_{k-1} \land A_{k}\land B,& \quad &G^{v_{k-1}}(X_k,X_1+\dots+X_k) \leftarrow  A_k\land B
 \end{alignat*}
 The modified query is $Q'\leftarrow \bigotimes_i A_i \otimes B \otimes \bigotimes_vG^v$
 and after factorization we have
\begin{align*}
    &Q'^{v_1}(X_1+X_2)\leftarrow A_1\otimes A_2\otimes G^{v_1},\quad \quad  Q'^{v_i}(X_1+\dots +X_{i+1})\leftarrow A_{i+1} \otimes Q'^{v_{i-1}}\otimes G^{v_i},\\
    &Q'^{v_{k-1}}()\leftarrow A_k\otimes B \otimes  Q'^{v_{k-2}}\otimes G^{v_{k-1}}
\end{align*}
 This leads to an algorithm with runtime $O(|\adom(D)|^{\ceil{(k+1)/2}})$ where the dominating part is the processing of the node $v_{\ceil{(k-1)/2}}$ and the evaluation of the reduct $G^{v_{\ceil{(k-1)/2}}}$.
 Note that this is precisely the optimal runtime for the $(k+1)$-SUM often assumed.
\end{example}
In general, $\ctw(Q) \leq \atw(Q)$, but the runtimes of Theorem~\ref{thm:algorithmIntegers} and~\ref{thm:adom} are incomparable, because the basis of the runtimes are ordered in the opposite direction: $n(D) \geq |\adom(D)|$.  However, we notice that, if $Q$ is a traditional CQ, then $\ctw(Q)=\atw(Q)$ and both agree with the traditional treewidth of $Q$ because every minimal CTD is equivalent to a traditional TD.
\subsection{Evaluation Algorithm with respect to the Database Size}
\label{sec:algorithm:databseSize}
Lastly, we describe an algorithm that uses the same idea as employed in Sec.~\ref{sec:algorithm:activeDomain}, but uses Theorem~\ref{thm:sparse} to evaluate the reducts.  
\begin{definition}[Fractional Convolution Hypertree Width]\label{def:fctw}
  Let $(T,\chi)$ be a CTD of the CSPQ $Q$.  Then, the \textit{fractional convolution hypertree width} of $(T,\chi)$ and $Q$ are defined as
    $$\fctw(T,\chi)=\max_{v\in V(T)}\rho^*(\chi(v)), \quad \fctw(Q)=\min_{(T,\chi)}\fctw(T,\chi)$$
\end{definition}
{Fractional convolution hypertree width measures each bag by the fractional edge cover number of its span rather than by its dimension. The resulting width captures the cost of constructing a sparse guard for each bag: the guard can be computed by the sparse algorithm, and the largest bag cost dominates the factorized evaluation. Thus, this is the CSPQ analogue of fractional hypertree width.}
\begin{restatable}{theorem}{thmfractionaltw}
\label{thm:fractional}
  Let $Q$ be a CSPQ.  Then, $\sem{Q}$ can be computed in time $O(|D|^{\fctw(Q)})$.
\end{restatable}
\begin{proof}[Proof Sketch]
  We proceed similar to in the proof of Theorem~\ref{thm:adom} but define the reducts differently.  Let $(T,\chi)$ be a width optimal CTD.  Fix a node $v\in T(V)$ and denote $\bm C^v_0 := \chi(v)$.  By Lemma~\ref{lemma:rhostar:geq} we can construct a reduct $G^v(\bm C^v_0) \leftarrow \bigwedge_{i=1,m}P_i(\bm C_i)$ such that $\rho^*_{G^v}(\bm C^v_0) = \rho^*(\bm C_0^v)$.  We evaluate $G^v$ in time $O(|D|^{\rho^*(\chi(v))})$ using Theorem~\ref{thm:dense}.  Then, we continue as in the proof of Theorem~\ref{thm:adom} verbatim.  The total runtime is $O(|D|^{\fctw(Q)})$ because it is dominated by the evaluation of the reducts.
\end{proof}
We sketch the algorithm on an example:
\begin{example}
    Let
    $Q(Y+Z)\leftarrow R_1(W+X,X+Y) \otimes R_2(W,X) \otimes R_3(X,Y)\otimes R_4(Y,Z) \otimes R_5(X+Y,Y+Z).$
    Further, consider the following minimal CTD with  width $\fctw=\frac{3}{2}$:
    \begin{center}
    \begin{tikzpicture}[scale = 0.9, font=\footnotesize]
        \node[rectangle, draw] (ABCDE) at (0,0) {
        $W,X,Y$
        };
        \node[rectangle, draw] (DEFGH) at (3,0.15) {
        $X,Y,Z$
        };
        \draw (ABCDE) -- (DEFGH) node[midway, above] {$X,Y$};
    \end{tikzpicture}
    \end{center} 
    To cover the left bag $\chi(c)$, we can use the fractional edge cover $(\frac{1}{2}, \frac{1}{2}, \frac{1}{2}, 0, 0)$ while we use $(0,0,\frac{1}{2},\frac{1}{2},\frac{1}{2})$ for the right bag $\chi(r)$.
    Then, to compute $\sem{Q}$, we would add the reducts
    \begin{align*}
        G^c(W,X,Y)\leftarrow R_1 \land R_2 \land R_3, \quad G^r(X,Y,Z)\leftarrow R_3\land R_4 \land R_5 ,
    \end{align*}
    to $Q$.
    Then, Theorem~\ref{thm:correctness} applied to $Q'\leftarrow \bigotimes_i R_i\otimes G^c \otimes G^r$ gives us the guarded queries
    \begin{align*}
        Q^c(X,Y)\leftarrow R_1 \otimes R_2 \otimes R_3 \otimes G^c, \quad  
        Q^r(Y+Z)\leftarrow R_4 \otimes R_5 \otimes Q^c  \otimes G^r.
    \end{align*}
    Each of these four queries $G^c,G^r,Q^c,Q^r$ can be solved in time $O(|D|^{3/2})$ and, $\sem{Q}=\sem{Q^r}$.
\end{example}
For traditional CQs $Q$, the value $\fctw(Q)$ coincides with the fractional tree width $fhtw(Q)$. Further, our width notions are invariant under variable substitutions.
For any CSPQ $Q$ and substitution $\phi$:
$$\ctw(Q)=\ctw(Q^\phi), \quad \atw(Q)=\atw(Q^\phi), \quad \fctw(Q)=\fctw(Q^\phi)$$

%% file: 7-other-fields.tex
\section{CSPQs over other Fields}
\label{sec:otherFields}
Up to this point, we always considered databases and CSPQs where the underlying field of domain elements were the rationals $\Q$.
However, this was only relevant in Theorem~\ref{thm:algorithmIntegers} and was actually never used anywhere else.
Thus, all other statements hold equally when $\Q$ is replaced by an arbitrary field $\field$.
However, 
one can also give an algorithm whose runtime has $\ctw(Q)+1$ as its exponent for fields with positive characteristic\footnote{Recall, a field $\field$ has characteristic $p> 0$ when adding $\1$ $p$-times gives $\0$ while adding $\1$ less times does not give $\0$.} $p>0$.
The idea of the algorithm in Section~\ref{sec:algorithm:integers} was to have a runtime of $O(n^{\ctw(T,\chi)+1})$ where $n$ is an upper bound on the number of field elements that may appear in one of the relations in the body of the query $Q$ or in one of the messages $\sem{Q^c}$.
With this in mind, let us denote the field elements used in $Q$ by $\adom(Q)\subseteq \field$ (i.e., the coefficients in $\bm A_0,\dots ,\bm A_m$)
and write $\adom(D,Q) :=\adom(D)\cup \adom(Q)$.
All field elements used by the messages must be derived from $\adom(D,Q)$.
Hence, if $\adom(D, Q)$ is contained in some finite subfield of $\field$ of size $n$, then also the messages will use only these $n$ field elements: We denote with $n(D,Q)$ the size of the smallest subfield containing $\adom(D,Q)$.
\begin{restatable}{theorem}{thmrelevantsubfield}
\label{thm:relevantSubfield}
For fields with positive characteristic, $Q$ can be evaluated in time $O(n(D,Q)^{\ctw(Q)+1})$.
\end{restatable}

%% file: 8-conclusion.tex
\section{Conclusion and Future Work}

\label{sec:conslusions}

In this paper, we extended sum-product queries to allow linear expressions over variables, resulting in convolution sum-product queries (CSPQs). We then adapted worst-case optimal join algorithms to this setting. Further, we used convolution tree-decompositions to factorize CSPQs and presented three width measures that bound the complexity of CSPQ evaluation relative to different notions of input size. We then extended this to queries over arbitrary fields.

{Several directions suggest themselves for future work. A natural first step is to identify further problems that can be expressed directly as CSPQs and to study how the techniques developed here perform on them. Such examples could clarify which features of a problem are captured by the three width measures, and could suggest additional structural parameters or specialized evaluation methods.}

{Exact-weight problems provide a particularly interesting test case for this program. The exact-weight subgraph problems studied in~\cite{DBLP:conf/icalp/AbboudL13} can be expressed as CSPQs. However, our  current algorithms  do not always  match the conditional lower bounds proven in~\cite{DBLP:conf/icalp/AbboudL13} for these problems. Matching these bounds may require incorporating additional information, such as functional dependencies, into the evaluation framework and the definition of width.}
{For example, the 4-path exact-weight problem in~\cite{DBLP:conf/icalp/AbboudL13} asks whether a graph with $n$ nodes and weighted edges has a path $E(A,B),E(B,C),E(C,D)$ whose weights sum to exactly 0.  Theorem 2 in~\cite{DBLP:conf/icalp/AbboudL13} proves that, under the $k$-SUM conjecture, this is not possible in time $O(n^{3-\varepsilon})$.  It turns out to be difficult to model exact sum queries as either SPQs or CSQPs.  For example, to write this query as an SPQ we need a semiring $(\Q,\oplus, \otimes, \bm 0, \bm 1)$ where $\otimes = +$ (because we need to add the three weights $E(A,B)+E(B,C)+E(C,D)$), but how should we define $\oplus$?  We need an operator such that $x_1 \oplus x_2 \oplus \cdots$ distinguishes between the cases when $0 \in \set{x_1, x_2, \ldots}$ and $0 \not\in \set{x_1, x_2, \ldots}$, and at the same time $+$ (which is our $\otimes$) distributes over $\oplus$, i.e. $(x_1 \oplus x_2 \oplus \cdots) + y = (x_1 + y) \oplus (x_2 + y) \oplus \cdots$. No such operator seems to exists.  Another option for the 4-path query is to express it as a CSQP over the Boolean semiring, similarly to our 3SUM query: $Q()\leftarrow E(A,B,X)\wedge E(B,C,Y) \wedge E(C,D,-X-Y)$, where the third argument of $E$ carries the weight of the edge. Then, a valid Convolution Tree Decomposition is $\{A,B,C,X,Y\}-\{C,D,-X-Y\}$. However, to evaluate the ``size'' of the bag, we have to use the fact that the following functional dependencies hold: $A,B\rightarrow X,\;\; B,C\rightarrow Y,\;\; C,D\rightarrow -X-Y$. In that case, we can assert that the bags have size 3 and the algorithm runs in time $O(n^3)$.  This suggests that it may be possible to extend CTDs to evaluate this query in time $O(n^3)$, matching the lower bound in~\cite{DBLP:conf/icalp/AbboudL13}, but this requires new techniques like handling functional dependencies.  We leave this for future work.}

{On the algorithmic side, non-combinatorial techniques may further improve CSPQ evaluation over semirings such as the standard $(+,\cdot)$ semiring and the Boolean $(\lor,\land)$ semiring. Fast matrix multiplication and fast Fourier transforms are two examples of techniques that may apply, depending on the semiring and the structure of the query. A systematic characterization of when such techniques can be used for CSPQ evaluation would be valuable. Another direction is to develop a CSPQ analogue of $(\#)$-submodular width.}

{On the complexity-theoretic side, conditional lower bounds for CSPQs could be derived from conjectures such as $(\min,+)$-convolution and $k$-SUM. Conversely, hardness conjectures for CSPQs might yield tighter lower bounds for restricted classes of SPQs.}

{Finally, one could extend the algebraic foundations beyond linear algebra over fields. This may lead to query languages with non-linear expressions or with domain elements from non-field algebraic structures such as groups, together with new notions of decomposition and width suited to those settings.}

\section*{Acknowledgements}

The work of Merkl was partially supported by the Vienna Science and Technology Fund (WWTF) [10.47379/ICT2201].
Suciu was partially supported by NSF IIS 2314527 and NSF 2507117 III.

%% file: AA-3.tex
\section{Full Proofs for Section~\ref{sec:cspq}}
\label{app:sec3}
Here, we give the missing/complete proofs of the lemmas of  Section~\ref{sec:cspq}
\lemsound*
\begin{proof}
    We have to show two things:
    1) We have to show that the infinite sum on the right hand side of Equation~\eqref{eq:semantic} is well-defined, i.e., that it only consists of finitely many non-$\0$ terms.
    2) We have to show that the support of $\sem{Q}$ is finite.
    Fix a CSPQ $Q(\bm A_0)\leftarrow \bigotimes_{i=1,m}R_i(\bm A_i)$ over $\Q^d$ with active space $\activespace$.
    To prove 1), let us fix a $\bm x\in \Q^{|\bm A_0|}$ and assume towards a contradiction that infinitely many distinct $u_1,\dots \in \activespace$ such that for all $j$
    $$u_j\in \activespace\colon \bm x = \inner{\bm A_0}{u_j}, \quad \text{ and } \quad \bigotimes_{i=1,m} R_i(\inner{\bm A_i}{u_j})\neq \0.$$
    This implies that for any $i,j$ combination also $R_i(\inner{\bm A_i}{u_j})\neq \0$.
    However, the support of each $R_i$ is finite.
    Thus, there must exist some $u_{j_1}$ and $u_{j_2}$ such that for all $i$
    $$\inner{\bm A_i}{u_{j_1}} = \inner{\bm A_i}{u_{j_2}}.$$
    Consequently, $\inner{a}{u_{j_1}-u_{j_2}}=0$ for any $a\in \bigcup_{i=1,m}\bm A_i$.
    I.e., $u_{j_1}-u_{j_2}=0\in \Q^d/(\bigcup_{i=1,m}\bm A_i)^\bot$.
    Put differently, $u_{j_1}=u_{j_2}$, contradicting the assumption that $u_{j_1}$ and $u_{j_2}$ are distinct.
    Thus, $\sem{Q}\colon \Q^{|\bm A_0|}\rightarrow \sring$ is well-defined.
    For 2) assume towards a contradiction that infinitely many distinct $\bm x_1,\dots\in \Q^{|\bm A_0|}$ such that for all $j$
    $$\sem{Q}(\bm x_j)\neq \0.$$
    Thus there are also $u_1,\dots\in \activespace$ such that 
    $$\bm x_j= \inner{\bm A_0}{u_j}, \quad \text{ and }\quad \forall i=1,m \colon R_i(\inner{\bm A_i}{u_j})\neq \0.$$
    Again, the support of each $R_i$ is finite.
    Thus, there must exist some $u_{j_1}$ and $u_{j_2}$ such that for all $i$
    $$\inner{\bm A_i}{u_{j_1}} = \inner{\bm A_i}{u_{j_2}}.$$
    Then, also 
    $$\bm x_{j_1}=\inner{\bm A_0}{u_{j_1}} = \inner{\bm A_0}{u_{j_2}} =  \bm x_{j_2}$$
    This contradict the assumption that $\bm x_{j_1}$ and $\bm x_{j_2}$ are distinct.
\end{proof}
\lemmasubstitution*
\begin{proof}
  Denote $S_1 := \Span(Q_1)$, $S_2 := \Span(Q_2)$ and notice that $S_2 = \varphi(S_1)$.  Let $\activespace_1 :=V_1/S_1^\perp, \activespace_2:=V_2/S_2^\perp$ be the active spaces of $Q_1, Q_2$ respectively.  We claim that there exists a linear, injective function $\varphi^* : V_2/S_2^\perp \rightarrow V_1/S_1^\perp$ with the following property:\footnote{When both null spaces are $S_1^\perp=\set{0}$, $S_2^\perp=\set{0}$ then $\varphi^*$ is the \emph{adjoint} of $\varphi$ (the transposed matrix).}
  \begin{align}
    \forall v_1 \in V_1, \ \ \ \ \forall [u_2] \in V_2/S_2^\perp: && \inner{\varphi(v_1)}{[u_2]} = & \inner{v_1}{\varphi^*([u_2])}\label{eq:varphistar}
  \end{align}
  The claim suffices to prove the lemma, because $\varphi^*$ is also a bijection (since $\dim(V_2/S_2^\perp)= \dim(V_1/S_1^\perp)$, because $\dim(V_i/S_i^\perp)=\dim(S_i)$ for $i=1,2$, and $\dim(S_2)=\dim(\varphi(S_1))=\dim(S_1)$); then $\forall \bm x \in \Q^{|\bm A_0|}$,
    \begin{align*}
      \sem{Q_1}(\bm x)&=  \bigoplus_{[v_1]\in \activespace_1\colon \inner{\bm A_0}{[v_1]}=\bm x}\bigotimes_iR^D_i(\inner{\bm A_i}{[v_1]}) \\
      &= \bigoplus_{[v_2]\in \activespace_2\colon \inner{\bm A_0}{\varphi^*([v_2])}=\bm x}\bigotimes_iR^D_i(\inner{\bm A_i}{\varphi^*([v_2])})  \\
      &=   \bigoplus_{[v_2]\in \activespace_2\colon \inner{\varphi(\bm A_0)}{[v_2]}=\bm x}\bigotimes_iR^D_i(\inner{\varphi(\bm A_i)}{[v_2]})\\
      &=\bigoplus_{[v_2]\in \activespace_2\colon \inner{\bm B_0}{[v_2]}=\bm x}\bigotimes_iR^D_i(\inner{\bm B_i}{[v_2]}) \\
      &=\sem{Q_2}(\bm x)
  \end{align*}
  It remains to prove the existence of $\varphi^*$.  Fix a vector $v_2 \in V_2$, and consider the function $f : S_1 \rightarrow \Q$, given by $f(v_1) := \inner{\varphi(v_1)}{[v_2]}$.  Since this is a linear function, there exists a vector $z \in V_1$ such that $f(v_1)=\inner{v_1}{z}$ for all $v_1 \in V_1$.  Define $\varphi^*([v_2]):= [z] \in V_1/S_1$. The function $\varphi^*:\activespace_2\rightarrow \activespace_1$ is well defined, because if $\inner{v_1}{[z]}=\inner{v_1}{[z']}$ for all $v_1 \in S_1$, then $\inner{v_1}{z-z'}=0$ for all $v_1 \in S_1$, implying $z-z' \in S_1^\perp$, hence $[z]=[z']$.  This proves the existence of $\varphi^*$ satisfying~\eqref{eq:varphistar}; linearity of $\varphi^*$ follows immediately from the linearity of $\varphi$ and of the dot product.  It remains to prove that it is injective.  If $\varphi^*([v_2])=\varphi^*([v_2'])$ then $\inner{\varphi(v_1)}{[v_2]}=\inner{\varphi(v_1)}{[v_2']}$ for all $v_1 \in S_1$, which, by definition, implies $\inner{\varphi(v_1)}{v_2}=\inner{\varphi(v_1)}{v_2'}$, i.e. $v_2-v_2' \in S_2^\perp$, hence $[v_2]=[v_2']$.   This proves the claim, and completes the proof of the lemma.
\end{proof}

%% file: AA-4.tex
\section{Full Proofs for Section~\ref{sec:basicAlgorithm}}
Before giving the missing proofs of Section~\ref{sec:basicAlgorithm}, we formally state and proof that 1) we can evaluate convolution-free CSPQs using traditional conjunctive query evaluation algorithms and 2) we can compute a reduct $P(\bm C)$ of $R(\bm A)$ in linear time:
\begin{lemma}
    Let $Q(\bm A_0)\leftarrow \bigotimes_{i=1,m}R_i(\bm A_i)$ be a convolution-free CSPQ (that is, $\bigcup_{i=1,m}\bm A_i=\{e_1,\dots,e_d\}$).
    Further, let $C(X_1,\dots,X_d)\leftarrow \bigotimes_{i=1,m}R_i(\bm E_i)$ be the full SPQ with the same body as $Q$ but where $e_j$ is replaced by the corresponding variable $X_j$.
    Then, $\sem{Q}$ can be computed in time $O(|\adom(D)|^d)$ and time $O(|D|^{\rho^*(C)})$
\end{lemma}
\begin{proof}
    First notice that the active space of $Q$ is $\activespace = \Q^d/\{e_1,\dots,e_d\}^\perp=\Q^d/\{0\}$.
    Thus, clearly, there is a bijection between valuations $v\colon \{X_1,\dots,X_d\}\rightarrow \Q$ and elements $[(v_1,\dots,v_d)]\in \activespace$ by equating $v_j=v(X_j)$ for all $j=1,d$.
    Furthermore, under this bijection
    $$\bigotimes_{i=1,m}R_i(\inner{\bm A_i}{[(v_1,\dots,v_d)]}) = \bigotimes_{i=1,m}R_i(v(\bm E_i))$$
    and the non-$\0$ term are precisely the support of $\sem{C}$.
    Thus, to evaluate $Q$ is suffices to go over the $(v_1,\dots,v_d)\in \sem{C}$ group them by $\inner{\bm A_0}{[(v_1,\dots,v_d)]}$ and aggregate via $\oplus$.
    That is,
    $$\sem{Q}(\bm x)=\bigoplus_{(v_1,\dots,v_d)\in \sem{C}}\bigotimes_{i=1,m}R_i(\inner{\bm A_i}{[(v_1,\dots,v_d)]}).$$
    As $C$ can be evaluated in time $O(|\adom(D)|^d)$ and time $O(|D|^{\rho^*(C)})$, also $Q$ can be evaluated in the same time frame.
\end{proof}
\begin{lemma}
\label{lem:computeReduct}
    Let $P(\bm C)$ be a reduct of $R(\bm A)$.
    Then, we can compute a reduct $P^D$, i.e., a Boolean $\sring$-relation $P^D$ such that $$\forall v\in \Q^d\colon R^D(\inner{\bm A}{v})\neq \bm 0 \Rightarrow P^{D}(\inner{\bm C}{v})=\bm 1,$$ in time $O(|R^D|)$.
\end{lemma}
\begin{proof}
Since $\bm C\subseteq \Span(\bm A)$ we can represent each vector $c_i \in \bm C$ as a linear combination
$c_i = \sum_{j=1,r} \lambda_{ij}a_j$ of the vectors in $\bm A= (a_1, \ldots, a_r)$.  
Then, let us define $P^D$ by iterating over each tuple $(u_1, \ldots, u_r) \in \supp(R^D)$ and setting $(\sum_{j=1,r} \lambda_{1j}u_j, \ldots, \sum_{j=1,r} \lambda_{pj}u_j)\mapsto \1$ in $P^D$.
All other tuples $\Q^{|\bm C|}$ are not in $\supp(P^D)$.
Then, clearly, $|P^D| \leq |R^D|$ and this process only takes time $O(|R^D|)$.  
Now lets fix a $v\in \Q^d$ be such that $R^D(\inner{\bm A}{v})\neq \0$.
Then, $\inner{\bm A}{v}=(u_1,\dots,u_r)\in \supp(R^D)$ and $(\sum_{j=1,r} \lambda_{1j}u_j, \ldots, \sum_{j=1,r} \lambda_{pj}u_j)\mapsto \1$ in $P^D$.
Further, notice that
\begin{align*}
\inner{\bm C}{v} & =\inner{(c_1,\dots,c_d)}{v}    \\
& =(\inner{c_1}{v},\dots,\inner{c_d}{v} )    \\
& =(\sum_{j=1,r}\lambda_{1,j}\inner{a_j}{v},\dots,\sum_{j=1,r}\lambda_{d,j}\inner{a_j}{v}) \\
& =(\sum_{j=1,r}\lambda_{1,j}u_j,\dots,\sum_{j=1,r}\lambda_{d,j}u_j).
\end{align*}
Thus, $P^{D}(\inner{\bm C}{v})=\bm 1$ as required.
\end{proof}
\lemmaguardone*
\begin{proof}
    We only have to show for any $v\in \Q^d$ that
    $$\bigotimes_{i=1,m}R_i(\inner{\bm A_i}{v}) = \bigotimes_{i=1,m}R_i(\inner{\bm A_i}{v}) \otimes \sem{G}(\inner{\bm C_0}{v}).$$
    To that end let us assume $\bigotimes_{i=1,m}R_i(\inner{\bm A_i}{v})\neq \0$ as equality otherwise trivially holds.
    Then, for every $i=1,m$ also $R_i(\inner{\bm A_i}{v})\neq \0$ and due to the definition of a reduct for every $j=1,n$ we have $P_j(\inner{\bm C_j}{v}) = \1$.
    Thus, $\bigwedge_{j=1,m}P_j(\inner{\bm C_j}{[v]}) = \1$ where $[v]\in \activespace_G$ and $\activespace_G$ is the active space of~$G$.
    Thus, moreover, $\sem{G}(\inner{\bm C_0}{[v]}) = \sem{G}(\inner{\bm C_0}{v}) = \1$.
    I.e.,
    $$\bigotimes_{i=1,m}R_i(\inner{\bm A_i}{v}) \otimes \sem{G}(\inner{\bm C_0}{v}) = \bigotimes_{i=1,m}R_i(\inner{\bm A_i}{v}) \otimes \1= \bigotimes_{i=1,m}R_i(\inner{\bm A_i}{v}) $$
    as required.
\end{proof}
\lemmaguardtwo*
\begin{proof}
  In the definition of $\sem{Q}^D$ given by Eq.~\eqref{eq:semantic}, we can restrict the summation to vectors $v \in \activespace = \Q^d/(\Span(Q))^\perp$ that satisfy $G(\inner{\bm C_0}{v})\neq \bm 0$. 
  Further, notice that for any tuple of vectors $\bm I$ the map $\inner{\bm I}{-}\colon \Q^d / \bm I^\perp \rightarrow \Q^{|I|}$ is injective.
  Thus, as $\Span(\bm C_0)=\Span(Q)$, we have $\Q^d/\bm C_0^\perp = \activespace$
  and 
  $\inner{\bm C_0}{-}\colon \activespace \rightarrow \Q^{|\bm C_0|}$ is injective.
  Computing the (partial) inverse map $\phi\colon \Q^{|\bm C_0|} \rightarrow  \activespace$
  is simple by applying basic linear algebra and expressable by a matrix.
  Thus, computing the $v\in \activespace$ such that $\inner{\bm C_0}{v}\neq \0$ is a matter of going over $\bm y\in \supp(G)$ and applying $\phi(\bm y)$.
  I.e., we can restrict the $\leq |G|$ many elements $\phi(\supp(G))\subseteq \activespace$.
\end{proof}

%% file: AA-5.tex
\section{Full Proofs for Section~\ref{sec:decomposition}}
Here, we give the missing/complete proofs for  Section~\ref{sec:decomposition}
\lemmafactorization*
\begin{proof}
    First we show that $Q_1$ is well-defined.
    That is, that $\bm A_0\subseteq \Span \bigcup_{i=1,m} \bm A_i \cup \bm B_0$.
    To that end, take $c\in \bm A_0$ which we write as $c=a+b\in \Span \bigcup_{i=1,m} \bm A_i\cup \bigcup_{j=1,n} \bm B_j \supseteq \bm B_0 $ where $a\in \Span \bigcup_{i=1,m}\bm A_i, b\in \Span \bigcup_{j=1,n} \bm B_j$.
    Thus, also, $b=c-a\in \Span\bigcup_{i=0,m} \bm A_i$ and consequently, $b\in \Span \bm B_0$.
    Therefore, $c=a+b\in \Span \bigcup_i \bm A_i \cup \bm B_0$.
    For convenience we write $V_1:=\Span(Q_1)=\Span \bm B_0 \cup \bigcup_{i=1,m} \bm A_i=\Span \bigcup_{i=0,m}\bm A_i, d_1:=\dim V_1$, and $V_2:=\Span(Q_2)= \Span \bigcup_{j=1,n} \bm B_j, d_2=\dim V_2$, and $V_0=\Span \bm B_0, d_0=\dim V_0$.
    Notice that $V_0=V_1\cap V_2$.
    Then, due to the Substitution Lemma~\ref{lemma:substitution}, we can assume
    $$V_1 = \Span \{e_1,\dots, e_{d_1}\}, \quad V_0= \Span \{e_{d-d_2+1}, \dots, e_{d-d_2+d_0}\}, \quad  V_2 = \Span \{e_{d-d_2+1},\dots, e_{d}\}.$$
    Note that $d+d_0\geq d_1+d_2$.
    To achieve this form, we can start with a basis $B_0$ of $V_0=V_1\cap V_2$ using the vectors $\bm B_0$, then extend it with elements $B_1$ from $V_1$ to a basis $B_0\cup B_1$ of $V_1$, and then extend $B_0$ again with elements $B_2$ from $V_2$ to a basis $B_0\cup B_2$ of $V_2$.
    I.e., in the end, 
    $$\Span B_0 = V_0=V_1\cap V_2,\quad \Span B_0\cup B_1 = V_1,\quad B_0 \cup B_2 = V_2.$$
    Now, for $\bm x \in \Q^{|\bm A_0|}$ we have
    \begin{align*}
        \sem{Q}(\bm x)=\bigoplus_{[u]\in \Q^d/\{0\}\colon \inner{\bm A_0}{u}=\bm x}  \bigotimes_{i=1,m} R_i(\inner{\bm A_i}{u})\otimes \bigotimes_{j=1,n} S_j(\inner{\bm B_j}{u}).
    \end{align*}
    We can uniquely break up $u$ into $u=u_1+u_0+u_2$ where 
    $$u_1\in \Span \{e_1,\dots, e_{d-d_2}\}, \quad u_0\in V_0, \quad u_2\in \Span \{e_{d-d_2+d_0+1},\dots, e_d\}.$$
    Thus, $u_1+u_0\in V_1$ and $u_0+u_2\in V_2$.
    Note that $\inner{\bm B_j}{u_1}=0$ and $\inner{\bm A_i}{u_2}=0$ for any $j=1,n,i=1,m$, as well as $\inner{\bm A_0}{u_2}=0$ and $\inner{\bm B_0}{u_1}=0=\inner{\bm B_0}{u_2}$.
    Therefore, further,
    \begin{align*}
        \sem{Q}(\bm x)& =\bigoplus_{u=u_1+u_1+u_2\in \Q^d\colon \inner{\bm A_0}{u_1+u_0}=\bm x}  \bigotimes_{i=1,m} R_i(\inner{\bm A_i}{u_1+u_0})\otimes \bigotimes_{j=1,n} S_j(\inner{\bm B_j}{u_0+u_2}) \\
        & =\bigoplus_{u_1+u_0 \colon \inner{\bm A_0}{u_1+u_0}=\bm x}  \bigotimes_i R_i(\inner{\bm A_i}{u_1+u_0})\otimes \bigoplus_{u_2}\bigotimes_j S_j(\inner{\bm B_j}{u_0+u_2}).
    \end{align*}
    Now let us consider a $\bm y\in \Q^{|\bm B_0|}$ such that $\bm y = \inner{\bm B_0}{u_0}$ and compute
    $$\sem{Q_2}(\bm y) = \bigoplus_{v\in \Q^d/V_2^\bot \colon \inner{\bm B_0}{v}=\bm y}\bigotimes_j S_j(\inner{\bm B_j}{v}).$$
    Note that we can equivalently write $v\in  \Q^d/V_2^\bot$ as $v=[v_0+v_2]_{V_2^\bot}$ for some unique vectors $v_0\in V_0, v_2\in \Span \{e_{d-d_2+d_0+1},\dots, e_d\}$.
    Further, note that $\inner{\bm B_0}{v_0+v_2}=\inner{\bm B_0}{v_0}=\bm y$ and, thus, $v_0=u_0$.
    Hence, the only freedom is in varying $v_2$ and 
    $$\sem{Q_2}(\bm y) = \bigoplus_{v_2}\bigotimes_j S_j(\inner{\bm B_j}{v_0+v_2}).$$
    Plugged into the previous equation
    $$\sem{Q}(\bm x) = \bigoplus_{u_1+u_0\colon \inner{\bm A_0}{u_1+u_0}=\bm x}  \bigotimes_i R_i(\inner{\bm A_i}{u_1+u_0}) \otimes \sem{Q_2}(\inner{\bm B_0}{u_0}).$$
    Further, as $\inner{\bm B_0}{u_1}=0$ we get
    $$\sem{Q}(\bm x) = \bigoplus_{u_1+u_0\colon \inner{\bm A_0}{u_1+u_0}=\bm x}  \bigotimes_i R_i(\inner{\bm A_i}{u_1+u_0}) \otimes \sem{Q_2}(\inner{\bm  B_0}{u_1+u_0}).$$
    Then, as $V_1^\perp = \Span \{e_{d-d_2+d_0+1},\dots, e_d\}$, we can equivalently sum over $[u_1+u_0] \in Q^d/V_1^\perp$.
    I.e., 
    \begin{align*}
        \sem{Q}(\bm x) & = \bigoplus_{[u_1+u_0]\in \Q^d/V_1^\bot\colon \inner{\bm A_0}{[u_1+u_0]}=\bm x}  \bigotimes_i R_i(\inner{\bm A_i}{[u_1+u_0]}) \otimes \sem{Q_2}(\inner{\bm B_0}{[u_1+u_0]}) \\
        & = \sem{Q_1}(\bm x).
    \end{align*}
    This completes the proof.
\end{proof}
\lemnice*
\begin{proof}
    To prove this, we simply start from $(T,\chi)$ and then, pick $\chi'(v)$ such that it is a basis of vector space on the right hand side of Equation~\eqref{eq:niceTD}.
    I.e.,
    $$\Span \chi'(v) = \Span \zeta(v) + \sum_{u \text{ is adjacent to } v}(\Span \zeta(T_u)\cap \Span \zeta(T_v))$$
    Then, we claim that $(T,\chi')$ is a minimal CTD as required.
    To that end, first notice that for any $v\in V(T)$ and subtree $T'$ of $T$
    $$\Span \zeta(v)\subseteq \Span \chi(v), \quad \Span \zeta(T')\subseteq \Span \chi(T').$$
    Thus, for any edge $uv$ in $T$ we have
    $$\Span \zeta(T_u)\cap \Span \zeta(T_v) \subseteq \Span \chi(T_u)\cap \Span \chi(T_v)= \Span \chi(u)\cap \Span \chi(v)\subseteq \Span \chi(v).$$
    Putting these fact together, ensures that $\Span \chi'(v)\subseteq \Span \chi(v)$ for any $v\in V(T)$.
    Thus, it only remains to show that $(T,\chi')$ is a minimal CTD of $Q$, i.e., properties (1)--(3) are satisfied.
    Properties (1), and (3) are satisfied by construction.
    For (1), notice that $\Span\zeta(v) \subseteq \Span \chi'(v)$.
    Hence, we are only left with property (2), the connectedness condition.
    To that end, let $uv$ be an edge in $T$.
    We claim
    $$\Span\chi'(T_u)=\Span \zeta(T_u), \quad \Span\chi'(T_v)=\Span \zeta(T_v).$$
    We only need to show ``$\subseteq$'' as ``$\supseteq$'' from the point-wise inclusion $\Span\zeta(v) \subseteq \Span \chi'(v)$.
    To the and, let $T'=T_v$ (the same works for $T'=T_u$).
    Then, to show that $\Span \chi'(w)$ is contained in $\Span \zeta(T')$ for every $w\in V(T')$, it suffices to show that all subspaces in the sum in Equation~\eqref{eq:niceTD} also lie in $\Span \zeta(T')$.
    Clearly, $\Span \zeta(w)\subseteq \Span \zeta(T')$.
    Then, for any edge $e$ incident to $w$, one of the subtrees of $T-e$ must be a subtree of $T'$ as well.
    For this subtree $T''$, $\Span \zeta(T'')\subseteq \Span \zeta(T')$ as $V(T'')\subseteq V(T')$.
    Consequently, $\Span \chi'(w)\subseteq \Span \zeta(T')$.
    As $w\in V(T')$ was arbitrary, $\Span \chi'(T')\subseteq \Span \zeta(T')$.
    Now, by the construction of $\chi'(u)$ and $\chi'(v)$, we have
    $$\Span \chi'(T_u) \cap \Span \chi'(T_v) = \Span \zeta(T_u) \cap \Span \zeta(T_v) \subseteq \Span \chi'(u) \cap \Span \chi'(v).$$
    Of course, the reverse inclusion holds trivially.
    Hence, we have shown property (2) of a CTD.
    Everything put together, we can assert that $(T,\chi')$ is a minimal CTD of $Q$ of the required form.
\end{proof}
Actually, in the proof we have even seen the following:
\begin{lemma}
    \label{lem:zetachi}
    Let $(T,\chi)$ be a minimal CTD of $Q$.
    Then, for any $uv\in E$, we have $\Span \chi(T_u)=\Span \zeta(T_u)$ and $\Span \chi(T_v)=\Span \zeta(T_v)$.
\end{lemma}
\begin{proof}
    Recall that in the proof of Lemma~\ref{lem:nice} we have shown this for $\chi'$.
    However, because here $(T,\chi)$ is already minimal, $\Span \chi'(w)=\Span \chi(w)$ for any $w\in V(T)$.
\end{proof}
\thmcorrectness*
\begin{proof}
    We prove this by induction on the size of the CTD $(T,\chi)$.
    When $T$ only consists of a single node $r$, then the statement trivially holds as $\bm B^v = \bm A_0$ and every body atom of $Q$ is covered by $r$.
    Equally, $Q^v$ simply is of the same form as $Q$:
    $$Q^v(\bm A_0)\leftarrow \bigotimes_i R_i(\bm A_i)$$
    Now let the root $r$ of $T$ at least contain some child $c$.
    Further, let $T_c,T_r$ be the two connected components of $T-rc$ where $c$ is the root of $T_c$ and $r$ is the root of $T_r$.
    Then, let 
    \begin{align*}
        I_c&:=\{i=1,m\mid R_i(\bm A_i) \text{ is covered by some } v\in V(T_c)\},\\ 
        I_r&:=\{i=1,m\mid R_i(\bm A_i) \text{ is covered by some } v\in V(T_r)\}.
    \end{align*}
    Then, let us apply the Factorization Lemma~\ref{lemma:factorization} to
    $$Q(\bm A_0)\leftarrow \bigotimes_{i\in I_r}R_i(\bm A_i) \otimes \bigotimes_{i\in I_c}R_i(\bm A_i) $$
    where we use $\bm B_0=\bm B^c$.
    We show that $\bm B_0=\bm B^c$ valid set of vectors.
    To that end, notice that due to Lemma~\ref{lem:zetachi}
    \begin{align*}
        (\Span \bm A_0\cup \bigcup_{i\in I_r}\bm A_i)\cap (\Span \bigcup_{i\in I_c}\bm A_i)& = \Span \chi(T_r)\cap \Span \chi(T_c) \\
        & = \Span \chi(r)\cap \Span \chi(c) \\
        & = \Span \bm B^c.
    \end{align*}
    Thus, we can apply the Factorization Lemma~\ref{lemma:factorization}.
    That is, for 
    $$Q_1(\bm A_0)\leftarrow \bigotimes_{i\in I_r}R_i(\bm A_i)\otimes Q_2(\bm B^c), \quad Q_2(\bm B^c)\leftarrow \bigotimes_{i\in I_c}R_i(\bm A_i)$$
    we have $\sem{Q}=\sem{Q_1}$.
    Then, we claim that $(T_c,\chi)$ is a minimal CTD of $Q_2$
    while $(T_r,\chi)$ is a minimal CTD of $Q_1$
    where we use $r\in V(T_r)$ to cover $Q_2(\bm B^c)$.
    First, notice that due to Lemma~\ref{lem:zetachi}  for $v\in V(T_r)$ we have $\chi(v)\subseteq \chi(T_r) = \zeta(T_r)=\Span \bm A_0 \cup \bigcup_{i\in I_r}\bm A_i = \Span \bm B^c \cup \bigcup_{i\in I_r}\bm A_i$, and  for $v\in V(T_c)$ we have $\chi(v)\subseteq \chi(T_c) = \zeta(T_c)=\Span \bigcup_{i\in I_c}\bm A_i$.
    Thus, $(T_c,\chi), (T_r,\chi)$ are at least candidates for being CTDs.
    Then, property (1) is clearly satisfied.
    Further, for property (2), simply notice that for any edge $uv\in T_c$ or $uv\in T_r$, the subtree $T_u,T_v$ to consider are smaller than they would be in $T$.
    Hence, the ``$\subseteq$'' is easier to satisfy in $T_c$ and $T_r$ than in $T$.
    I.e., the ``$\subseteq$'' relation holds in $T_c$ and $T_r$ as it does in $T$.
    The ``$\supseteq$'' relation holds trivially.
    Consequently,  $(T_c,\chi)$ and $ (T_r,\chi)$ are in fact CTDs.
    Now, let us consider property (4).    
    To that end, let $\zeta_c,\zeta_r$ be the function $\zeta$ of Definition~\ref{def:nice} for $(T_c,\chi)$ and $ (T_r,\chi)$, respectively.
    Now, let $vu\in T$ be an edge $vu\neq rc$ where $v$ is the parent of $u$.
    Let us assume that $vu$ is in $T_r$ and let us consider the two subtrees $T_v,T_u$ of $T-vu$ and the two subtrees $T_{rv},T_{ru}$ of $T_r-vu$.
    Then, we claim
    $$\Span\zeta(T_v)\cap \Span\zeta(T_u)=\Span\zeta_r(T_{rv})\cap \Span\zeta_r(T_{ru})$$
    First, notice that $T_u=T_{ru}$ and $\zeta_r(T_{ru}) = \zeta(T_{u})$.
    For $\Span\zeta_r(T_{rv})$, let us construct a basis of $\Span \bigcup_{i=1,m} \bm A_i = \Span \zeta(T_r)+\Span \zeta(T_c)$ in the following way (Due to Lemma~\ref{lem:zetachi}, the intersection is $\Span \bm B^c$):
    We start with a basis $B_{rc}\subseteq \bm B^c$ that spans $\Span \bm B^c$, and extend it with elements $B_c\subseteq \Span \zeta(T_c)$ to a basis $B_{rc}\cup B_c$ of $\Span \zeta(T_c)$.
    At the same time, we extend $B_{rc}$ with $B_{rv}\subseteq \Span \zeta_r(T_{rv})$ to a basis $B_{rc}\cup B_{rv}$ of $\Span \zeta_r(T_{rv})$ and, further, with $B_{r}\subseteq \Span \zeta_r(T_{r})$ to a basis $B_{rc}\cup B_{rv}\cup B_{r}$ of $\Span \zeta_r(T_{r})$.
    Then, $\Span\zeta(T_v)\subseteq \Span (B_{rc}\cup B_c \cup B_{rv})$ while $\Span\zeta(T_u)\subseteq \Span (B_{rc}\cup B_{rv}\cup B_{r})$.
    Thus, 
    $$\Span\zeta(T_v)\cap \Span\zeta(T_u)\subseteq \Span (B_{rc}\cup B_{rv}).$$
    At the same time 
    $$\Span\zeta(T_v) \cap \Span (B_{rc}\cup B_{rv}) = \Span \zeta_r(T_{rv}).$$
    Thus, 
    $$\Span\zeta(T_v)\cap \Span\zeta(T_u)=\Span\zeta(T_v)\cap  \Span (B_{rc}\cup B_{rv})\cap \Span\zeta(T_u)=\Span\zeta_r(T_{rv})\cap \Span\zeta_r(T_{ru})$$
    as claimed.
    We never used the fact that $r$ is the parent of $c$, hence, the analog statement holds for $vu$ in $T_c$ (simply switch the roles of $r$ and $c$).
    Further, $\zeta_r(w)=\zeta(w)$ for $w\in V(T_r), w\neq r$ as well as $\zeta_c(w)=\zeta(w)$ for $w\in V(T_c), w\neq c$.
    Thus, property (3) is verified for all but $w=r$ or $w=c$ as in Equation~\eqref{eq:niceTD} we can simply replace all $\zeta$ with $\zeta_r$ (resp. $\zeta_c$) as well as $T_u, T_v$ with $T_{ru}, T_{rv}$ (resp. $T_{cu}, T_{cv}$).
    For $w=r$ or $w=c$ we simply shift the term $\Span \bm B^c = \Span \zeta(T_r)\cap \Span \zeta(T_c)$ from the big sum to the left.
    That is, notice $\Span \zeta(w) + \Span \bm B^c= \Span \zeta_r(w)$ (resp. $\Span \zeta(w) + \Span \bm B^c= \Span \zeta_c(w)$).
    All other terms are taken care of as before.
    Thus, we conclude that $(T_c,\chi), (T_r,\chi)$ are in fact minimal CTDs.
    Thus, by induction, we assert that $\sem{Q_2}=\sem{Q^c}$.
    Further, let
    $$Q'_1(\bm A_0)\leftarrow \bigotimes_{i\in I_r}R_i(\bm A_i)\otimes Q^c(\bm B^c).$$
    Then, clearly, $\sem{Q}=\sem{Q_1}=\sem{Q'_1}$.
    Now we can apply the induction also to $Q'_1$
    and get $Q''_1$ of the form
    $$Q''_1(\bm A_0)\leftarrow \bigotimes_{r \text{ covers } R_i(\bm A_i)} R_i(\bm A_i) \otimes Q^c(\bm B^c)\otimes \bigotimes_{c'\in child(v)} Q^{c'}(\bm B^{c'})$$
    where $c'$ are the children of $r$ in $T_r$, i.e., this excludes $c$.
    However, note that $Q''_1$ is syntactically the same as $Q^r$.
    Thus, $\sem{Q}=\sem{Q''_1}=\sem{Q^r}$.
\end{proof}

%% file: AA-6.tex
\section{Full Proofs for Section~\ref{sec:algorithm}}
\label{app:arxiv1}
In this section, we give the full proofs of Theorems~\ref{thm:algorithmIntegers},~\ref{thm:adom}, and~\ref{thm:fractional} but also prove informal statements that $\ctw(Q)=\atw(Q)=\tw(Q)$ for traditional SPQs $Q$, and that $\ctw(Q)=\ctw(Q^\phi), \; \atw(Q)=\atw(Q^\phi), \; \fctw(Q)=\fctw(Q^\phi)$ for CSPQs $Q$ and substitutions $\phi$.
\thmalgorithmIntegers*
\begin{proof}[Proof Sketch]
    Due to Lemma~\ref{lem:nice}, we can assume, w.l.o.g, that there exists a minimal CTD $(T,\chi)$ such that $\ctw(T,\chi)=\ctw(Q)$.
    We evaluate $Q$ bottom-up, using the Query Factorization Theorem~\ref{thm:correctness}.
    We pick the $\bm B^v$ bottom-up in the following way:
    Fix a $v\in V(T)$ where $v$ is not the root of $T$ and all the $\bm B^c$ for children $c$ of $v$ have already been picked. 
    Then, let $\bm D^v = \{d_1,\dots,d_{l_v}\}$ be basis of $\Span \chi(v)$ consisting of vectors that appear in the body of $Q^v$.
    Further, let $b^v_1,\dots,b^v_{k_v}$ be a basis of $\Span \chi(v) \cap \Span \chi(parent(v))$
    Thus, we can express $b^v_1,\dots,b^v_{k_v}$ as a linear combinations 
    $$b^v_1 = \sum_j \lambda_{1,j}d_j,\quad \dots,\quad b^v_{k_v} = \sum_j \lambda_{k_v,j}d_j.$$
    Now, we can multiply all of the $b^v_1,\dots,b^v_{k_v}$ with a big enough integer $\alpha\in \Z$ such that all the factors $\alpha \lambda_{i,j}$ are also integers.
    Therefore, $\bm B^v:=(\alpha b^v_1,\dots,\alpha b^v_{k_v})$ is then a basis of $\Span \chi(v) \cap \Span \chi(parent(v))$ where all basis elements are expressable as an \textit{integral} combination of the vectors in $\bm D^v$.
    Hence, with the $\bm B^v$ picked as above and when the supports of the relations in the body of $Q^v$ (including the messages) only use integers which are at most of size $n$, then, due to Lemma~\ref{lemma:substitution}, also $\sem{Q^v}$ only contains integers and their size is bound by $\alpha^{l_v} \cdot n$.
    Therefore, for any $v\in V(T)$ (now $v$ is allowed to be the root), the active domain of the relations in the body of $Q^v$ are integers and bounded by $O(n(D))$ as we can hide the factors $\alpha^{l_v}$ in the big $O$ notation.
    Thus, due to Theorem~\ref{thm:dense}, we can evaluate every $Q^v$ in time $O(n(D)^{l_v})$ where $\max_v l_v = \ctw(Q)+1$.
    Thus, this algorithm runs in time $O(n(D)^{\ctw(Q)+1})$.
\end{proof}
\thmadom*
\begin{proof}
  Fix a width optimal minimal CTD $(T,\chi)$ of $Q(\bm A_0)\leftarrow \bigotimes_{i=1,m}R_i(\bm A_i)$.  We combine the results of Sec.~\ref{sec:basicAlgorithm} and Sec.~\ref{sec:decomposition} to obtain this result.  The idea is to first add reducts to $Q$ such that after applying the Factorization Theorem~\ref{thm:correctness} that all subqueries $Q^v$ are guarded.  Fix a node $v \in V(T)$, denote $\bm C^v_0 := \chi(v)$, and let $\bm C=(b_1,\dots,b_k)$ be a a vector cover of $\bm C^v_0$ of size $k:=\alpha(\bm C^v_0)$.  For each $b_j$, let $R_i(\bm A_i)$ be an atom such that $b_j\in \bm A_i$.  
  We compute the reducts $P_j(b_j)$ of $R_i(\bm A_i)$, and denote with $G^v(\bm C^v_0)\leftarrow \bigwedge_{j=1,k} P_j(b_j)$ a reduct of $Q$.  
  By Theorem~\ref{thm:dense} $\sem{G^v}^D$ can be computed in time $O(|\adom(D)|^{\alpha(\chi(v))})$.  By Lemma~\ref{lemma:guard:1}, we can add an atom $G^v(\bm C^v_0)$ to $Q$ without affecting the semantics. 
  We repeat this process for all tree nodes $v\in V(T)$. 
    The modified query is
    $$Q_{\text{mod}}(\bm A_0)\leftarrow \bigotimes_{i=1,m}R_i(\bm A_i)\otimes \bigotimes_{v\in V(T)}G^v(\bm C^v_0)$$
  Then, $(T,\chi)$ is still valid minimal CTD of the modified query where we cover the new atoms $G^v(\bm C_0^v)$ each with the node $v$.
  We compute the query bottom-up on the tree $T$, using the Factorization Theorem~\ref{thm:correctness}.  
  That is, the queries are of the form
  $$Q_{\text{mod}}^v(\bm B^v)\leftarrow \bigotimes_{v \text{ covers } R_i(\bm A_i)}R_i(\bm A_i)\otimes \bigotimes_{c\in child(v)} Q_{\text{mod}}^c(\bm B^c) \otimes G^v(\bm C^v_0)$$
  where $\Span(Q_{\text{mod}}^v) = \Span \chi (v)=\Span(\bm C^v_0)$ by construction.
  Thus, each query $Q_{\text{mod}}^v$ can be computed in time $O(|\sem{G^v}|)$ by Lemma~\ref{lemma:guard:2}, because it is guarded by $G^v$.  Thus, the runtime is dominated by the time needed to compute reducts, which is $O(|\adom(D)|^{\atw(Q)+1})$.
    This completes the proof as for the root $r$ of $T$ we have $\sem{Q^r_{\text{mod}}}=\sem{Q}$.
\end{proof}
\thmfractionaltw*
\begin{proof}
    We proceed as in the proof of Theorem~\ref{thm:adom}.
  Fix a width optimal minimal CTD $(T,\chi)$ of $Q(\bm A_0)\leftarrow \bigotimes_{i=1,m}R_i(\bm A_i)$.  
  The idea is again to first add reducts to $Q$ such that after applying the Factorization Theorem~\ref{thm:correctness} that all subqueries $Q^v$ are guarded.  
  Fix a node $v \in V(T)$ and denote $\bm C^v_0 := \chi(v)$.
  By Lemma~\ref{lemma:rhostar:geq} we can construct a reduct $G^v(\bm C^v_0) \leftarrow \bigwedge_{i=1,m}P_i(\bm C_i)$ of $Q$ such that $\rho^*_{G^v}(\bm C^v_0) = \rho^*(\bm C_0^v)$.  
  We can evaluate $G^v$ in time $O(|D|^{\rho^*(\chi(v))})$ using Theorem~\ref{thm:dense}.  
  By Lemma~\ref{lemma:guard:1}, we can add an atom $G^v(\bm C^v_0)$ to $Q$ without affecting the semantics. 
  We repeat this process for all tree nodes $v\in V(T)$. 
    The modified query is again
    $$Q_{\text{mod}}(\bm A_0)\leftarrow \bigotimes_{i=1,m}R_i(\bm A_i)\otimes \bigotimes_{v\in V(T)}G^v(\bm C^v_0)$$
  Then, $(T,\chi)$ is still valid minimal CTD of the modified query where we cover the new atoms $G^v(\bm C_0^v)$ each with the node $v$.
  We compute the query bottom-up on the tree $T$, using the Factorization Theorem~\ref{thm:correctness}.  
  That is, the queries are of the form
  $$Q_{\text{mod}}^v(\bm B^v)\leftarrow \bigotimes_{v \text{ covers } R_i(\bm A_i)}R_i(\bm A_i)\otimes \bigotimes_{c\in child(v)} Q_{\text{mod}}^c(\bm B^c) \otimes G^v(\bm C^v_0)$$
  where $\Span(Q_{\text{mod}}^v) = \Span \chi (v)=\Span(\bm C^v_0)$ by construction.
  Thus, each query $Q_{\text{mod}}^v$ can be computed in time $O(|\sem{G^v}|)$ by Lemma~\ref{lemma:guard:2}, because it is guarded by $G^v$.  Thus, the runtime is dominated by the time needed to compute reducts, which is $O(|\adom(D)|^{\atw(Q)+1})$.
    This completes the proof as for the root $r$ of $T$ we have $\sem{Q^r_{\text{mod}}}=\sem{Q}$.
\end{proof}
\begin{theorem}
    Let $Q(\bm E_0)\leftarrow \bigotimes_{i=1,m}R_i(\bm E_i)$ be a convolution-free CSPQ with $\bm E_0\subseteq \{e_1,\dots,e_d\}$ and $C(\bm X_0)\leftarrow \bigwedge_{i=1,m}R_i(\bm X_i)$ the CQ obtained by replacing the $e_j$ with simply variables $X_j$.
    Then, 
    $$\ctw(Q)=\atw(Q)=\tw(C),\quad \text{and}\quad \fctw(Q)= \ftw(C)$$
    where $\tw(C)$ is the traditional treewidth of $C$ (with a single bag covering all output variables) and $\ftw(C)$ is the traditional fractional treewidth of $C$ (with a single bag covering all output variables).
\end{theorem}
\begin{proof}
    First notice that traditional tree decompositions can easily be interpreted as CTDs and the width of each bag seen as a CTD is the same as the width of the bag when seen a traditional tree decomposition (for $\atw$ note that a traditional bag is a vector cover of itself, and for $\fctw$ recall Lemma~\ref{lemma:rhostar:cspq:cq}).
    Thus, it suffices to prove that minimal CTD $(T,\chi)$ are such that $\chi(v)=\{e_{j_1},\dots,e_{j_{n_v}}\}$.
    To that end, let $(T,\chi)$ be an arbitrary minimal CTD of $Q$.
    Then, for any $v\in V(T)$ we simply have $$\zeta(v)=\bigcup_{R_i(\bm E_i) \text{ is covered by } v}\bm E_i=\{e_{i_1}, ,\dots,e_{i_{m_v}}\}$$
    Thus, by looking at the RHS of Equation~\ref{eq:niceTD} we notice that
    $$\Span \chi(v)= \Span \{e_{j_1},\dots,e_{j_{n_v}}\}$$
    as all vector spaces in the sum are of the same form.
    Therefore, w.l.o.g, we can pick $\chi(v)= \Span \{e_{j_1},\dots,e_{j_{n_v}}\}$ and $(T,\chi)$ can be seen as a traditional tree decomposition.
    This completes the proof.
\end{proof}
\begin{theorem}
    Let $Q$ be a CSPQ and $\phi$ a substitution. Then,
    $$\ctw(Q)=\ctw(Q^\phi), \quad \atw(Q)=\atw(Q^\phi), \quad \fctw(Q)=\fctw(Q^\phi).$$
\end{theorem}
\begin{proof}
    We simply have to observe that substitutions do not affect the width notions.
    To that end, let $(T,\chi)$ be a CTD of the CSPQ $Q$.
    Then, let us define $\chi^\phi\colon v\mapsto \phi(\chi(v))$.
    Notice that $\phi$ is a isomorphism from the span of $Q$ to the image of $\phi$.
    Thus, properties (1)--(3) of minimal tree decompositions are satisfied by $(T,\chi)$ and $Q$ if and only if there are satisfied by $(T,\chi^\phi)$ and $Q^\phi$.
    Furthermore, 1) $\dim \Span \chi(v)=\dim \Span \chi^\phi(v)$ so $\ctw(Q)=\ctw(Q^\phi)$; 2) $\bm C$ is a vector cover of $\bm C_0$ with respect to $Q$ if and only if $\phi(\bm C)$ is a vector cover of $\phi (\bm C_0)$ with respect to $Q^\phi$ so $\alpha_Q(\chi(v))=\alpha_{Q^\phi}(\chi^\phi(v))$ and $\atw(Q)=\atw(Q^\phi)$; 3) we have already noted in Lemma~\ref{lemma:same:rhostar} that $\rho^*_Q(\chi(v))=\rho^*_{Q^\phi}(\chi^\phi(v))$ so $\fctw(Q)=\fctw(Q^\phi)$.
\end{proof}

%% file: AA-7.tex
\section{Proofs for Section~\ref{sec:otherFields}}
\label{app:arxiv2}
For a field $\mathbb{F}$ with characteristic $p>0$ and some field elements $A\subseteq \mathbb{F}$, we write $\field_p(A)$ for the smallest subfield of $\mathbb{F}$ containing $A$. 
Before proving Theorem~\ref{thm:relevantSubfield}, we show that we can always ask for a tree decomposition such that the bags are made up of field elements from $\field_p(\adom(Q))$.
Thus, essentially, we never have to go beyond $\field_p(\adom(D, Q))$ where $|\field_p(\adom(D, Q))| = n(D,Q)$.
Note that $\field_p(\adom(D, Q))$ is infinite when $\adom(D,Q)$ contains an non-algebraic element.
In that case, we consider Theorem~\ref{thm:relevantSubfield} to be trivially true.
\begin{restatable}{lemma}{lemsubfield}
\label{lem:subfield}
    Let $(T,\chi)$ be a tree decomposition of the CSPQ $Q$ over a field of characteristic $p>0$.
    Then, there exists a minimal tree decomposition $(T,\chi')$ of $Q$ where $\chi'(v)\subseteq \mathbb{F}_{p}(\adom(Q))^d$ for all $v\in V(T)$ and $\ctw(T,\chi')\leq \ctw(T,\chi)$.
\end{restatable}
\begin{proof}
    Due to Lemma~\ref{lem:nice}, we know that there always exists a minimal tree decomposition $(T,\chi')$.
    Note that $\Span\chi'(v)\subseteq \Span \chi(v)$ for each $v\in V(T)$ implies $\ctw(T,\chi')\leq \ctw(T,\chi)$.
    Clearly, for every $v\in V(T)$ it does not matter which set of vectors $\chi'(v)$ are picked, as long as they satisfy Equation~\eqref{eq:niceTD} in Definition~\ref{def:nice}.
    Then, notice that each $\Span \zeta(v)$ is vector space with a basis in $\mathbb{F}_{p}(\adom(Q))^d$.
    Thus, as this property is closed under taking intersections and sums, $\Span \chi(v)$ is a vector space with a basis in $\mathbb{F}_{p}(\adom(Q))^d$.
    Thus, w.l.o.g., $\chi'(v)\subseteq \mathbb{F}_{p}(\adom(Q))^d$ for every $v\in V(T)$.
\end{proof}
Now, we can prove Theorem~\ref{thm:relevantSubfield}.
\thmrelevantsubfield*
\begin{proof}
    Due to Lemma~\ref{lem:subfield}, we can assume w.l.o.g.\ that there is a minimal tree decomposition $T,\chi$ such that $\ctw(T,\chi)=\ctw(Q)$ and $\chi(v)\subseteq\mathbb{F}_{p}(\adom(Q))^d$ for all $v\in V(T)$.
    Then, let us evaluate $Q$ according to Theorem~\ref{thm:correctness} via the subqueries $Q^v, v\in V(T)$.
    As $\chi(v)\subseteq\mathbb{F}_{p}(\adom(Q))^d$, we can also pick $\bm B^v\subseteq \mathbb{F}_{p}(\adom(Q))^d$.
    We show by induction that $\adom(\sem{Q^v})\subseteq \mathbb{F}_{p}(\adom(D, Q))$.
    To that end, we perform a variable substitution according to Lemma~\ref{lemma:substitution} such that $\phi^{-1}(e_1),\dots,\phi^{-1}(e_k)$ appear in the body of $Q^v$.
    Then, $\sem{Q^{v\phi}}(\inner{\phi(\bm B^c)}{u})\neq 0$
    only if $R_i(\inner{\phi(\bm A_i)}{u})\neq 0$ and $\sem{Q^c}(\inner{\phi(\bm B^c)}{u})\neq 0$.
    Thus, as $\phi^{-1}(e_1),\dots,\phi^{-1}(e_k)$ appear somewhere in the body, we have $u\in \mathbb{F}_{p}(\adom(D,Q))^d$.
    Notice that all vectors $b\in \bm B^v$ are a $\mathbb{F}_{p}(\adom(Q))$-linear combination of $\phi^{-1}(e_1),\dots,\phi^{-1}(e_k)$.
    Then, also $\inner{\phi(b)}{u}\in \mathbb{F}_{p}(\adom(D,Q))$.
    Consequently, $\adom(\sem{Q^v})\subseteq \mathbb{F}_{p}(\adom(D,Q))$.
    Now we simply have to apply Theorem~\ref{thm:dense} to evaluate each of the subqueries $Q^v$ in time $O(|\mathbb{F}_{p}(\adom(D,Q))|^{\ctw(T,\chi)+1})$ as required.
\end{proof}

%% file: AA-linearEquations.tex
\section{CSPQs and SPQs with Linear Equality Constraints}
\label{app:le}
In this section we outline how CSPQs semantically capture SPQ with added linear equation constraints (SPQLEs).
To that end, we give two ways CSPQs capture SPQLEs.
The first version is very direct but adds relations while in the second version the relations used by the SPQLEs are exactly the same as the one used by the CSPQs.
We begin by formally defining SPQLEs.
\begin{definition}
    Let $R_1(\bm X_1),\dots, R_m(\bm X_m)$ be atoms using the variables $\bm X=(X_1,\dots,X_d)^T$ and let $\bm B\in \Q^{n \times d}$ be a matrix.
    A \textit{sum-product query with linear equation constraints} (SPQLEs) is an expression of the form:
    \begin{align*}
        L(\bm X_0)\leftarrow R_1(\bm X_1)\otimes \dots\otimes  R_m(\bm X_m),\; \text{s.t. } \bm B\bm X = 0 
    \end{align*}
    We require the safety condition $\bm X_0\subseteq \{X_1,\dots,X_d\}=\bigcup_{i=1,m} \bm X_0$
    The answer of $L$ over an $\sring$-database instance $D$ is the $\sring$-relation\footnote{with $\bm B\bm X$ (resp. $\bm Bv(\bm X)$) we mean the matrix multiplication of the matrix $\bm B$ with the transposed vector $\bm X=(X_1,\dots,X_d)^T$ (resp. $v(\bm X)=(v(X_1),\dots,v(X_d))^T$). }:
    \begin{align}
        \forall \bm x \in \Q^{|\bm X_0|}\colon \quad \sem{L}(\bm x) := \bigoplus_{v\colon \bm X\rightarrow \Q, \text{ s.t. } v(\bm X_0)=\bm x \text{ and } \bm Bv(\bm X)=0} \bigotimes_{i=1,m}R_i(v(\bm X_i)) \label{eq:spqle}
    \end{align}
\end{definition}
Formally, in a SPQLE we group all linear equation constraint in a single equation $\bm B \bm X=0$.
That is $\bm B$ is a matrix with $n$ rows $\bm B= \left(\begin{matrix} - b_1-\\\vdots\\\ - b_n-\end{matrix}\right)$
and we require that the $n$ linear equation $b_1v(\bm X)=0, \dots, b_n v(\bm X)=0$ hold of the valuations $v$ used in the sum.
Thus, to transform $L$ into a CSPQ $Q$, we can simply add a Boolean relation $\1\colon \Q \rightarrow \sring$ with $\1(0)=\1$ and $\1(\lambda)=\0$ for $\lambda\neq 0$.
We show:
\begin{theorem}
    Let 
    $$L(\bm X_0)\leftarrow R_1(\bm X_1)\otimes \dots\otimes  R_m(\bm X_m),\; \text{s.t. } \bm B\bm X = 0, \quad \text{ where }\bm B= \left(\begin{matrix} -b_1-\\\vdots\\\ -b_n-\end{matrix}\right) $$ 
    be a SPQLE.
    Then, $L$ is equivalent to the CSPQ
    $$Q(\bm E_0)\leftarrow R_1(\bm E_1)\otimes \dots\otimes  R_m(\bm E_m) \otimes \1(b_1) \otimes \dots \otimes \1(b_n)$$
    where the $\bm E_i$ are simply tuples of canonical basis vectors obtained from $X_i$ by replacing the occurrences of the variable $X_j$ with the vector $e_j$.
    We means that over any database $D$ we have $\sem{L}^D=\sem{Q}^D$
\end{theorem}
\begin{proof}
    We simply compare the non-$\0$ terms given in the Equation~\eqref{eq:spqle} with the non-$\0$ terms given in the Equation~\eqref{eq:semantic} (the semantics for a CSPQ).
    To that end, notice that $\Span(Q)=\Q^d$ and we associate the valuations $v$ with the vectors $[(v(X_1),\dots, v(X_d))]\in Q^d/\{0\}$.
    If $v$ produces a non-$\0$ term in Equation~\eqref{eq:spqle}, then $\bm Bv(\bm X)=0$.
    Thus, in particular, $0=b_jv(\bm X)=\inner{b_j}{[(v(X_1),\dots, v(X_d))]}$
    Thus, for $u:=[(v(X_1),\dots, v(X_d))]$
    \begin{align*}
        \bigotimes_{i=1,m}R_i(\inner{\bm E_i}{u}) \otimes \bigotimes_{j=1,n} \1(\inner{b_j}{u}) & = \bigotimes_{i=1,m}R_i(\inner{\bm E_i}{u}) \otimes \1\\
        & = \bigotimes_{i=1,m}R_i(v(\bm X_i))
    \end{align*}
    Moreover, $v(\bm X_0)=\inner{\bm E_0}{u}$.
    Conversely, if $u:=[(u_1,\dots,u_d)]$ produces a non-$\0$ term in Equation~\eqref{eq:semantic}, then $\1(\inner{b_j}{u})=\1$ for every $j=1,n$.
    Thus, $\bm Bv(\bm X)=0$ for $v\colon X_i\mapsto u_i$ and again
    \begin{align*}
        \bigotimes_{i=1,m}R_i(\inner{\bm E_i}{u}) \otimes \bigotimes_{j=1,n} \1(\inner{b_j}{u}) = \bigotimes_{i=1,m}R_i(v(\bm X_i))
    \end{align*}
    Moreover, $v(\bm X_0)=\inner{\bm E_0}{u}$.
\end{proof}
However, we can even capture SPQLEs without needed the additional relation $\1$ and only use the relations $R_1,\dots,E_m$.
We show:
\begin{theorem}
    Let 
    $$L(\bm X_0)\leftarrow R_1(\bm X_1)\otimes \dots\otimes  R_m(\bm X_m),\; \text{s.t. } \bm B\bm X = 0 $$
    be a SPQLE.
    Then, $L$ is equivalent to the CSPQ
    $$Q(\bm A_0)\leftarrow R_1(\bm A_1)\otimes \dots\otimes  R_m(\bm A_m)$$
    for some tuples of vectors $\bm A_0,\dots, \bm A_m$.
    We means that over and database $D$ we have $\sem{L}^D=\sem{Q}^D$
\end{theorem}
\begin{proof}
    We start by reducing $\bm B$ into reduced row echelon form.
    That is, recall that after permuting variables, there is a matrix
    $$\bm B'=\left(\begin{matrix} 0 & 0\\ \bm C & \bm I_{d-k}\end{matrix}\right)$$
    such that $\bm Bv(\bm X)=0$ if and only if $\bm B'v(\bm X)=0$.
    Here $\bm I_{d-k}$ is the $d-k\times d-k$ identity matrix and $\bm C\in \Q^{k\times d-k}$ is an arbitrary matrix.
    Thus, we can partition the variables into $\bm X = (X_1,\dots,X_k,X_{k+1},\dots,X_d)$ and the solutions to the equation $\bm B'\bm X=0$ comes from picking valuations of $X_1,\dots, X_k$ arbitrarily and determining the valuations for $X_{k+1},\dots, X_d$ from $\bm B'$.
    Let us write $\bm C = \left(\begin{matrix} -c_{k+1}-\\\vdots\\\ -c_{d}-\end{matrix}\right)$.
    Then, in $\bm A_0,\dots, \bm A_m$ let us replace $X_1,\dots,X_k$ with $e_1,\dots,e_k\in \Q^k$ and $X_{k+1},\dots,X_d$ with $-c_{k+1},\dots, -c_d\in \Q^k$.
    Notice that $\Span(Q)=\Q^k$ due to the safety condition of SPQLEs and that for any valuations $v$ such that $\bm Bv(\bm X)=0$ and $u:=(v(X_1),\dots,v(X_k))$ we have 
    $$\inner{e_1}{u}=v(X_1), \dots,  \inner{e_k}{u}=v(X_1),\quad  \inner{-c_{k+1}}{u}=v(X_{k+1}), \dots, \inner{-c_{d}}{u}=v(X_{d}).$$
    We now compare the terms given in the Equation~\eqref{eq:spqle} with the terms given in the Equation~\eqref{eq:semantic} (the semantics for a CSPQ).
    We associate the valuations $v$ such that $\bm Bv(\bm X)=0$ with the vectors $u=[(v(X_1),\dots, v(X_k))]\in \Q^k/\{0\}$.
    Note that this is a bijection as the valuations for $X_{k+1},\dots,X_d$ can be determined from the valuations for $X_1,\dots,X_k$ via $v(X_j)=\inner{-c_{j}}{u}$.
    Then starting from and $u$ or $v$, we have 
    \begin{align*}
        \bigotimes_{i=1,m}R_i(\inner{\bm A_i}{u})  = \bigotimes_{i=1,m}R_i(v(\bm X_i))
    \end{align*}
    Thus, $\sem{L}=\sem{Q}$ as also $\inner{\bm A_0}{u}=v(\bm X_0)$.
\end{proof}

%% file: biblio.bib
@inproceedings{DBLP:conf/pods/GreenKT07,
  author       = {Todd J. Green and
                  Gregory Karvounarakis and
                  Val Tannen},
  editor       = {Leonid Libkin},
  title        = {Provenance semirings},
  booktitle    = {Proceedings of the Twenty-Sixth {ACM} {SIGACT-SIGMOD-SIGART} Symposium
                  on Principles of Database Systems, June 11-13, 2007, Beijing, China},
  pages        = {31--40},
  publisher    = {{ACM}},
  year         = {2007},
  url          = {https://doi.org/10.1145/1265530.1265535},
  doi          = {10.1145/1265530.1265535},
  bibsource    = {dblp computer science bibliography, https://dblp.org}
}

@inproceedings{DBLP:conf/icdt/Veldhuizen14,
  author       = {Todd L. Veldhuizen},
  editor       = {Nicole Schweikardt and
                  Vassilis Christophides and
                  Vincent Leroy},
  title        = {Triejoin: {A} Simple, Worst-Case Optimal Join Algorithm},
  booktitle    = {Proc. 17th International Conference on Database Theory (ICDT), Athens,
                  Greece, March 24-28, 2014},
  pages        = {96--106},
  publisher    = {OpenProceedings.org},
  year         = {2014},
  url          = {https://doi.org/10.5441/002/icdt.2014.13},
  doi          = {10.5441/002/ICDT.2014.13},
  bibsource    = {dblp computer science bibliography, https://dblp.org}
}

@inproceedings{DBLP:conf/pods/NgoPRR12,
  author       = {Hung Q. Ngo and
                  Ely Porat and
                  Christopher R{\'{e}} and
                  Atri Rudra},
  editor       = {Michael Benedikt and
                  Markus Kr{\"{o}}tzsch and
                  Maurizio Lenzerini},
  title        = {Worst-case optimal join algorithms: [extended abstract]},
  booktitle    = {Proceedings of the 31st {ACM} {SIGMOD-SIGACT-SIGART} Symposium on
                  Principles of Database Systems, {PODS} 2012, Scottsdale, AZ, USA,
                  May 20-24, 2012},
  pages        = {37--48},
  publisher    = {{ACM}},
  year         = {2012},
  url          = {https://doi.org/10.1145/2213556.2213565},
  doi          = {10.1145/2213556.2213565},
  bibsource    = {dblp computer science bibliography, https://dblp.org}
}

@article{DBLP:journals/corr/KhamisNRR15,
  author       = {Mahmoud Abo Khamis and
                  Hung Q. Ngo and
                  Atri Rudra},
  title        = {{FAQ:} Questions Asked Frequently},
  journal      = {CoRR},
  volume       = {abs/1504.04044},
  year         = {2015},
  url          = {http://arxiv.org/abs/1504.04044},
  eprinttype    = {arXiv},
  eprint       = {1504.04044},
  bibsource    = {dblp computer science bibliography, https://dblp.org}
}

@article{DBLP:journals/sigmod/NgoRR13,
  author       = {Hung Q. Ngo and
                  Christopher R{\'{e}} and
                  Atri Rudra},
  title        = {Skew strikes back: new developments in the theory of join algorithms},
  journal      = {{SIGMOD} Rec.},
  volume       = {42},
  number       = {4},
  pages        = {5--16},
  year         = {2013},
  url          = {https://doi.org/10.1145/2590989.2590991},
  doi          = {10.1145/2590989.2590991},
  bibsource    = {dblp computer science bibliography, https://dblp.org}
}

@article{DBLP:journals/talg/GroheM14,
  author       = {Martin Grohe and
                  D{\'{a}}niel Marx},
  title        = {Constraint Solving via Fractional Edge Covers},
  journal      = {{ACM} Trans. Algorithms},
  volume       = {11},
  number       = {1},
  pages        = {4:1--4:20},
  year         = {2014},
  url          = {https://doi.org/10.1145/2636918},
  doi          = {10.1145/2636918},
  bibsource    = {dblp computer science bibliography, https://dblp.org}
}

@article{DBLP:journals/jcss/GottlobLS02,
  author       = {Georg Gottlob and
                  Nicola Leone and
                  Francesco Scarcello},
  title        = {Hypertree Decompositions and Tractable Queries},
  journal      = {J. Comput. Syst. Sci.},
  volume       = {64},
  number       = {3},
  pages        = {579--627},
  year         = {2002},
  url          = {https://doi.org/10.1006/jcss.2001.1809},
  doi          = {10.1006/JCSS.2001.1809},
  bibsource    = {dblp computer science bibliography, https://dblp.org}
}

@inproceedings{DBLP:conf/focs/AtseriasGM08,
  author       = {Albert Atserias and
                  Martin Grohe and
                  D{\'{a}}niel Marx},
  title        = {Size Bounds and Query Plans for Relational Joins},
  booktitle    = {49th Annual {IEEE} Symposium on Foundations of Computer Science, {FOCS}
                  2008, October 25-28, 2008, Philadelphia, PA, {USA}},
  pages        = {739--748},
  publisher    = {{IEEE} Computer Society},
  year         = {2008},
  url          = {https://doi.org/10.1109/FOCS.2008.43},
  doi          = {10.1109/FOCS.2008.43},
  bibsource    = {dblp computer science bibliography, https://dblp.org}
}

@article{DBLP:journals/corr/abs-2404-04369,
  author       = {Karl Bringmann and
                  Egor Gorbachev},
  title        = {A Fine-grained Classification of Subquadratic Patterns for Subgraph
                  Listing and Friends},
  journal      = {CoRR},
  volume       = {abs/2404.04369},
  year         = {2024},
  url          = {https://doi.org/10.48550/arXiv.2404.04369},
  doi          = {10.48550/ARXIV.2404.04369},
  eprinttype    = {arXiv},
  eprint       = {2404.04369},
  bibsource    = {dblp computer science bibliography, https://dblp.org}
}

@inproceedings{DBLP:conf/pods/KhamisCM0NOS19,
  author       = {Mahmoud Abo Khamis and
                  Ryan R. Curtin and
                  Benjamin Moseley and
                  Hung Q. Ngo and
                  XuanLong Nguyen and
                  Dan Olteanu and
                  Maximilian Schleich},
  editor       = {Dan Suciu and
                  Sebastian Skritek and
                  Christoph Koch},
  title        = {On Functional Aggregate Queries with Additive Inequalities},
  booktitle    = {Proceedings of the 38th {ACM} {SIGMOD-SIGACT-SIGAI} Symposium on Principles
                  of Database Systems, {PODS} 2019, Amsterdam, The Netherlands, June
                  30 - July 5, 2019},
  pages        = {414--431},
  publisher    = {{ACM}},
  year         = {2019},
  url          = {https://doi.org/10.1145/3294052.3319694},
  doi          = {10.1145/3294052.3319694},
  bibsource    = {dblp computer science bibliography, https://dblp.org}
}

@article{DBLP:journals/talg/CyganMWW19,
  author       = {Marek Cygan and
                  Marcin Mucha and
                  Karol Wegrzycki and
                  Michal Wlodarczyk},
  title        = {On Problems Equivalent to (min, +)-Convolution},
  journal      = {{ACM} Trans. Algorithms},
  volume       = {15},
  number       = {1},
  pages        = {14:1--14:25},
  year         = {2019},
  url          = {https://doi.org/10.1145/3293465},
  doi          = {10.1145/3293465},
  bibsource    = {dblp computer science bibliography, https://dblp.org}
}

@article{DBLP:journals/tods/KhamisCMNNOS20,
  author       = {Mahmoud Abo Khamis and
                  Ryan R. Curtin and
                  Benjamin Moseley and
                  Hung Q. Ngo and
                  XuanLong Nguyen and
                  Dan Olteanu and
                  Maximilian Schleich},
  title        = {Functional Aggregate Queries with Additive Inequalities},
  journal      = {{ACM} Trans. Database Syst.},
  volume       = {45},
  number       = {4},
  pages        = {17:1--17:41},
  year         = {2020},
  url          = {https://doi.org/10.1145/3426865},
  doi          = {10.1145/3426865},
  bibsource    = {dblp computer science bibliography, https://dblp.org}
}

@article{DBLP:journals/taco/TavarageriHAKGU21,
  author       = {Sanket Tavarageri and
                  Alexander Heinecke and
                  Sasikanth Avancha and
                  Bharat Kaul and
                  Gagandeep Goyal and
                  Ramakrishna Upadrasta},
  title        = {PolyDL: Polyhedral Optimizations for Creation of High-performance
                  {DL} Primitives},
  journal      = {{ACM} Trans. Archit. Code Optim.},
  volume       = {18},
  number       = {1},
  pages        = {11:1--11:27},
  year         = {2021},
  url          = {https://doi.org/10.1145/3433103},
  doi          = {10.1145/3433103},
  bibsource    = {dblp computer science bibliography, https://dblp.org}
}

@inproceedings{DBLP:conf/cgo/BaghdadiRRSAZSK19,
  author       = {Riyadh Baghdadi and
                  Jessica Ray and
                  Malek Ben Romdhane and
                  Emanuele Del Sozzo and
                  Abdurrahman Akkas and
                  Yunming Zhang and
                  Patricia Suriana and
                  Shoaib Kamil and
                  Saman P. Amarasinghe},
  editor       = {Mahmut Taylan Kandemir and
                  Alexandra Jimborean and
                  Tipp Moseley},
  title        = {Tiramisu: {A} Polyhedral Compiler for Expressing Fast and Portable
                  Code},
  booktitle    = {{IEEE/ACM} International Symposium on Code Generation and Optimization,
                  {CGO} 2019, Washington, DC, USA, February 16-20, 2019},
  pages        = {193--205},
  publisher    = {{IEEE}},
  year         = {2019},
  url          = {https://doi.org/10.1109/CGO.2019.8661197},
  doi          = {10.1109/CGO.2019.8661197},
  bibsource    = {dblp computer science bibliography, https://dblp.org}
}

@inproceedings{DBLP:conf/osdi/ChenMJZYSCWHCGK18,
  author       = {Tianqi Chen and
                  Thierry Moreau and
                  Ziheng Jiang and
                  Lianmin Zheng and
                  Eddie Q. Yan and
                  Haichen Shen and
                  Meghan Cowan and
                  Leyuan Wang and
                  Yuwei Hu and
                  Luis Ceze and
                  Carlos Guestrin and
                  Arvind Krishnamurthy},
  editor       = {Andrea C. Arpaci{-}Dusseau and
                  Geoff Voelker},
  title        = {{TVM:} An Automated End-to-End Optimizing Compiler for Deep Learning},
  booktitle    = {13th {USENIX} Symposium on Operating Systems Design and Implementation,
                  {OSDI} 2018, Carlsbad, CA, USA, October 8-10, 2018},
  pages        = {578--594},
  publisher    = {{USENIX} Association},
  year         = {2018},
  url          = {https://www.usenix.org/conference/osdi18/presentation/chen},
  bibsource    = {dblp computer science bibliography, https://dblp.org}
}

@article{DBLP:journals/taco/VerdoolaegeJCGTC13,
  author       = {Sven Verdoolaege and
                  Juan Carlos Juega and
                  Albert Cohen and
                  Jos{\'{e}} Ignacio G{\'{o}}mez and
                  Christian Tenllado and
                  Francky Catthoor},
  title        = {Polyhedral parallel code generation for {CUDA}},
  journal      = {{ACM} Trans. Archit. Code Optim.},
  volume       = {9},
  number       = {4},
  pages        = {54:1--54:23},
  year         = {2013},
  url          = {https://doi.org/10.1145/2400682.2400713},
  doi          = {10.1145/2400682.2400713},
  bibsource    = {dblp computer science bibliography, https://dblp.org}
}

@inproceedings{DBLP:conf/sigmod/MoerkotteN08,
  author       = {Guido Moerkotte and
                  Thomas Neumann},
  editor       = {Jason Tsong{-}Li Wang},
  title        = {Dynamic programming strikes back},
  booktitle    = {Proceedings of the {ACM} {SIGMOD} International Conference on Management
                  of Data, {SIGMOD} 2008, Vancouver, BC, Canada, June 10-12, 2008},
  pages        = {539--552},
  publisher    = {{ACM}},
  year         = {2008},
  url          = {https://doi.org/10.1145/1376616.1376672},
  doi          = {10.1145/1376616.1376672},
  bibsource    = {dblp computer science bibliography, https://dblp.org}
}

@inproceedings{DBLP:conf/otm/KianiS05,
  author       = {Ali Kiani and
                  Nematollaah Shiri},
  editor       = {Robert Meersman and
                  Zahir Tari and
                  Mohand{-}Said Hacid and
                  John Mylopoulos and
                  Barbara Pernici and
                  {\"{O}}zalp Babaoglu and
                  Hans{-}Arno Jacobsen and
                  Joseph P. Loyall and
                  Michael Kifer and
                  Stefano Spaccapietra},
  title        = {Containment of Conjunctive Queries with Arithmetic Expressions},
  booktitle    = {On the Move to Meaningful Internet Systems 2005: CoopIS, DOA, and
                  ODBASE, {OTM} Confederated International Conferences CoopIS, DOA,
                  and {ODBASE} 2005, Agia Napa, Cyprus, October 31 - November 4, 2005,
                  Proceedings, Part {I}},
  series       = {Lecture Notes in Computer Science},
  pages        = {439--452},
  publisher    = {Springer},
  year         = {2005},
  url          = {https://doi.org/10.1007/11575771\_28},
  doi          = {10.1007/11575771\_28},
  bibsource    = {dblp computer science bibliography, https://dblp.org}
}

@inproceedings{DBLP:conf/issads/KianiS05,
  author       = {Ali Kiani and
                  Nematollaah Shiri},
  editor       = {F{\'{e}}lix F. Ramos Corchado and
                  Victor M. Larios{-}Rosillo and
                  Herwig Unger},
  title        = {A Framework for Information Integration with Uncertainty},
  booktitle    = {Advanced Distributed Systems: 5th International School and Symposium,
                  {ISSADS} 2005, Guadalajara, Mexico, January 24-28, 2005, Revised Selected
                  Papers},
  series       = {Lecture Notes in Computer Science},
  pages        = {194--206},
  publisher    = {Springer},
  year         = {2005},
  url          = {https://doi.org/10.1007/11533962\_17},
  doi          = {10.1007/11533962\_17},
  bibsource    = {dblp computer science bibliography, https://dblp.org}
}

@inproceedings{DBLP:conf/ihis/KianiS05,
  author       = {Ali Kiani and
                  Nematollaah Shiri},
  editor       = {Axel Hahn and
                  Sven Abels and
                  Liane Haak},
  title        = {Answering queries with arithmetic expressions in heterogeneous systems},
  booktitle    = {Proceedings of the first international {ACM} workshop on Interoperability
                  of Heterogeneous Information Systems (IHIS'05), {CIKM} Conference,
                  Bremen, Germany, November 4, 2005},
  pages        = {17--24},
  publisher    = {{ACM}},
  year         = {2005},
  url          = {https://doi.org/10.1145/1096967.1096972},
  doi          = {10.1145/1096967.1096972},
  bibsource    = {dblp computer science bibliography, https://dblp.org}
}

@article{DBLP:journals/tods/WangY26,
  author       = {Qichen Wang and
                  Ke Yi},
  title        = {Conjunctive Queries with Comparisons},
  journal      = {{ACM} Trans. Database Syst.},
  volume       = {51},
  number       = {2},
  pages        = {7:1--7:37},
  year         = {2026},
  url          = {https://doi.org/10.1145/3769424},
  doi          = {10.1145/3769424},
  bibsource    = {dblp computer science bibliography, https://dblp.org}
}

@article{DBLP:journals/sigmod/0001023,
  author       = {Qichen Wang and
                  Ke Yi},
  title        = {Conjunctive Queries with Comparisons},
  journal      = {{SIGMOD} Rec.},
  volume       = {52},
  number       = {1},
  pages        = {54--62},
  year         = {2023},
  url          = {https://doi.org/10.1145/3604437.3604450},
  doi          = {10.1145/3604437.3604450},
  bibsource    = {dblp computer science bibliography, https://dblp.org}
}

@article{DBLP:journals/jacm/Klug88,
  author       = {Anthony C. Klug},
  title        = {On conjunctive queries containing inequalities},
  journal      = {J. {ACM}},
  volume       = {35},
  number       = {1},
  pages        = {146--160},
  year         = {1988},
  url          = {https://doi.org/10.1145/42267.42273},
  doi          = {10.1145/42267.42273},
  bibsource    = {dblp computer science bibliography, https://dblp.org}
}

@inproceedings{DBLP:conf/adbis/BrisaboaHPP98,
  author       = {Nieves R. Brisaboa and
                  H{\'{e}}ctor J. Hern{\'{a}}ndez and
                  Jos{\'{e}} R. Param{\'{a}} and
                  Miguel R. Penabad},
  editor       = {Witold Litwin and
                  Tadeusz Morzy and
                  Gottfried Vossen},
  title        = {Containment of Conjunctive Queries with Built-in Predicates with Variables
                  and Constants over any Ordered Domain},
  booktitle    = {Advances in Databases and Information Systems, Second East European
                  Symposium, ADBIS'98, Poznan, Poland, September 7-10, 1998, Proceedings},
  series       = {Lecture Notes in Computer Science},
  pages        = {46--57},
  publisher    = {Springer},
  year         = {1998},
  url          = {https://doi.org/10.1007/BFb0057716},
  doi          = {10.1007/BFB0057716},
  bibsource    = {dblp computer science bibliography, https://dblp.org}
}

@article{DBLP:journals/jcss/IbarraS99,
  author       = {Oscar H. Ibarra and
                  Jianwen Su},
  title        = {A Technique for Proving Decidability of Containment and Equivalence
                  of Linear Constraint Queries},
  journal      = {J. Comput. Syst. Sci.},
  volume       = {59},
  number       = {1},
  pages        = {1--28},
  year         = {1999},
  url          = {https://doi.org/10.1006/jcss.1999.1624},
  doi          = {10.1006/JCSS.1999.1624},
  bibsource    = {dblp computer science bibliography, https://dblp.org}
}

@article{DBLP:journals/amai/Srivastava93,
  author       = {Divesh Srivastava},
  title        = {Subsumption and Indexing in Constraint Query Languages with Linear
                  Arithmetic Constraints},
  journal      = {Ann. Math. Artif. Intell.},
  volume       = {8},
  number       = {3-4},
  pages        = {315--343},
  year         = {1993},
  url          = {https://doi.org/10.1007/BF01530796},
  doi          = {10.1007/BF01530796},
  bibsource    = {dblp computer science bibliography, https://dblp.org}
}

@book{DBLP:books/sp/Kuper00,
  editor       = {Gabriel M. Kuper and
                  Leonid Libkin and
                  Jan Paredaens},
  title        = {Constraint Databases},
  publisher    = {Springer},
  year         = {2000},
  url          = {https://doi.org/10.1007/978-3-662-04031-7},
  doi          = {10.1007/978-3-662-04031-7},
  isbn         = {3-540-66151-4},
  bibsource    = {dblp computer science bibliography, https://dblp.org}
}

@inproceedings{DBLP:conf/edbt/AfratiLM04,
  author       = {Foto N. Afrati and
                  Chen Li and
                  Prasenjit Mitra},
  editor       = {Elisa Bertino and
                  Stavros Christodoulakis and
                  Dimitris Plexousakis and
                  Vassilis Christophides and
                  Manolis Koubarakis and
                  Klemens B{\"{o}}hm and
                  Elena Ferrari},
  title        = {On Containment of Conjunctive Queries with Arithmetic Comparisons},
  booktitle    = {Advances in Database Technology - {EDBT} 2004, 9th International Conference
                  on Extending Database Technology, Heraklion, Crete, Greece, March
                  14-18, 2004, Proceedings},
  series       = {Lecture Notes in Computer Science},
  pages        = {459--476},
  publisher    = {Springer},
  year         = {2004},
  url          = {https://doi.org/10.1007/978-3-540-24741-8\_27},
  doi          = {10.1007/978-3-540-24741-8\_27},
  bibsource    = {dblp computer science bibliography, https://dblp.org}
}

@article{DBLP:journals/tcs/AfratiLM06,
  author       = {Foto N. Afrati and
                  Chen Li and
                  Prasenjit Mitra},
  title        = {Rewriting queries using views in the presence of arithmetic comparisons},
  journal      = {Theor. Comput. Sci.},
  volume       = {368},
  number       = {1-2},
  pages        = {88--123},
  year         = {2006},
  url          = {https://doi.org/10.1016/j.tcs.2006.08.020},
  doi          = {10.1016/J.TCS.2006.08.020},
  bibsource    = {dblp computer science bibliography, https://dblp.org}
}

@article{DBLP:journals/mst/KoutrisMRS17,
  author       = {Paraschos Koutris and
                  Tova Milo and
                  Sudeepa Roy and
                  Dan Suciu},
  title        = {Answering Conjunctive Queries with Inequalities},
  journal      = {Theory Comput. Syst.},
  volume       = {61},
  number       = {1},
  pages        = {2--30},
  year         = {2017},
  url          = {https://doi.org/10.1007/s00224-016-9684-2},
  doi          = {10.1007/S00224-016-9684-2},
  bibsource    = {dblp computer science bibliography, https://dblp.org}
}

@inproceedings{DBLP:conf/icdt/KhamisNOS19,
  author       = {Mahmoud Abo Khamis and
                  Hung Q. Ngo and
                  Dan Olteanu and
                  Dan Suciu},
  editor       = {Pablo Barcel{\'{o}} and
                  Marco Calautti},
  title        = {Boolean Tensor Decomposition for Conjunctive Queries with Negation},
  booktitle    = {22nd International Conference on Database Theory, {ICDT} 2019, Lisbon,
                  Portugal, March 26-28, 2019},
  series       = {LIPIcs},
  pages        = {21:1--21:19},
  publisher    = {Schloss Dagstuhl - Leibniz-Zentrum f{\"{u}}r Informatik},
  year         = {2019},
  url          = {https://doi.org/10.4230/LIPIcs.ICDT.2019.21},
  doi          = {10.4230/LIPICS.ICDT.2019.21},
  bibsource    = {dblp computer science bibliography, https://dblp.org}
}

@inproceedings{DBLP:conf/icalp/AbboudL13,
  author       = {Amir Abboud and
                  Kevin Lewi},
  editor       = {Fedor V. Fomin and
                  Rusins Freivalds and
                  Marta Z. Kwiatkowska and
                  David Peleg},
  title        = {Exact Weight Subgraphs and the k-Sum Conjecture},
  booktitle    = {Automata, Languages, and Programming - 40th International Colloquium,
                  {ICALP} 2013, Riga, Latvia, July 8-12, 2013, Proceedings, Part {I}},
  series       = {Lecture Notes in Computer Science},
  pages        = {1--12},
  publisher    = {Springer},
  year         = {2013},
  url          = {https://doi.org/10.1007/978-3-642-39206-1\_1},
  doi          = {10.1007/978-3-642-39206-1\_1},
  bibsource    = {dblp computer science bibliography, https://dblp.org}
}

@inproceedings{DBLP:conf/icalp/LincolnWWW16,
  author       = {Andrea Lincoln and
                  Virginia {Vassilevska Williams} and
                  Joshua R. Wang and
                  R. Ryan Williams},
  editor       = {Ioannis Chatzigiannakis and
                  Michael Mitzenmacher and
                  Yuval Rabani and
                  Davide Sangiorgi},
  title        = {Deterministic Time-Space Trade-Offs for k-SUM},
  booktitle    = {43rd International Colloquium on Automata, Languages, and Programming,
                  {ICALP} 2016, Rome, Italy, July 11-15, 2016},
  series       = {LIPIcs},
  pages        = {58:1--58:14},
  publisher    = {Schloss Dagstuhl - Leibniz-Zentrum f{\"{u}}r Informatik},
  year         = {2016},
  url          = {https://doi.org/10.4230/LIPIcs.ICALP.2016.58},
  doi          = {10.4230/LIPICS.ICALP.2016.58},
  bibsource    = {dblp computer science bibliography, https://dblp.org}
}

@inproceedings{DBLP:conf/pldi/Ragan-KelleyBAPDA13,
  author       = {Jonathan Ragan{-}Kelley and
                  Connelly Barnes and
                  Andrew Adams and
                  Sylvain Paris and
                  Fr{\'{e}}do Durand and
                  Saman P. Amarasinghe},
  editor       = {Hans{-}Juergen Boehm and
                  Cormac Flanagan},
  title        = {Halide: a language and compiler for optimizing parallelism, locality,
                  and recomputation in image processing pipelines},
  booktitle    = {{ACM} {SIGPLAN} Conference on Programming Language Design and Implementation,
                  {PLDI} '13, Seattle, WA, USA, June 16-19, 2013},
  pages        = {519--530},
  publisher    = {{ACM}},
  year         = {2013},
  url          = {https://doi.org/10.1145/2491956.2462176},
  doi          = {10.1145/2491956.2462176},
  bibsource    = {dblp computer science bibliography, https://dblp.org}
}

@inproceedings{DBLP:conf/nips/KrizhevskySH12,
  author       = {Alex Krizhevsky and
                  Ilya Sutskever and
                  Geoffrey E. Hinton},
  editor       = {Peter L. Bartlett and
                  Fernando C. N. Pereira and
                  Christopher J. C. Burges and
                  L{\'{e}}on Bottou and
                  Kilian Q. Weinberger},
  title        = {ImageNet Classification with Deep Convolutional Neural Networks},
  booktitle    = {Advances in Neural Information Processing Systems 25: 26th Annual
                  Conference on Neural Information Processing Systems 2012. Proceedings
                  of a meeting held December 3-6, 2012, Lake Tahoe, Nevada, United States},
  pages        = {1106--1114},
  year         = {2012},
  url          = {https://proceedings.neurips.cc/paper/2012/hash/c399862d3b9d6b76c8436e924a68c45b-Abstract.html},
  bibsource    = {dblp computer science bibliography, https://dblp.org}
}

@inproceedings{DBLP:conf/sc/DattaMVWCOPSY08,
  author       = {Kaushik Datta and
                  Mark Murphy and
                  Vasily Volkov and
                  Samuel Williams and
                  Jonathan Carter and
                  Leonid Oliker and
                  David A. Patterson and
                  John Shalf and
                  Katherine A. Yelick},
  title        = {Stencil computation optimization and auto-tuning on state-of-the-art
                  multicore architectures},
  booktitle    = {Proceedings of the {ACM/IEEE} Conference on High Performance Computing,
                  {SC} 2008, November 15-21, 2008, Austin, Texas, {USA}},
  pages        = {4},
  publisher    = {{IEEE/ACM}},
  year         = {2008},
  url          = {https://doi.org/10.1109/SC.2008.5222004},
  doi          = {10.1109/SC.2008.5222004},
  bibsource    = {dblp computer science bibliography, https://dblp.org}
}

@inproceedings{DBLP:conf/ics/HolewinskiPS12,
  author       = {Justin Holewinski and
                  Louis{-}No{\"{e}}l Pouchet and
                  P. Sadayappan},
  editor       = {Utpal Banerjee and
                  Kyle A. Gallivan and
                  Gianfranco Bilardi and
                  Manolis Katevenis},
  title        = {High-performance code generation for stencil computations on {GPU}
                  architectures},
  booktitle    = {International Conference on Supercomputing, ICS'12, Venice, Italy,
                  June 25-29, 2012},
  pages        = {311--320},
  publisher    = {{ACM}},
  year         = {2012},
  url          = {https://doi.org/10.1145/2304576.2304619},
  doi          = {10.1145/2304576.2304619},
  bibsource    = {dblp computer science bibliography, https://dblp.org}
}

@inproceedings{getreuer2009contour,
  title={Contour stencils for edge-adaptive image interpolation},
  author={Getreuer, Pascal},
  booktitle={Visual Communications and Image Processing 2009},
  volume={7257},
  pages={394--406},
  year={2009},
  organization={SPIE}
}

@article{DBLP:journals/toms/LuporiniLLKWHYK20,
  author       = {Fabio Luporini and
                  Mathias Louboutin and
                  Michael Lange and
                  Navjot Kukreja and
                  Philipp A. Witte and
                  Jan H{\"{u}}ckelheim and
                  Charles Yount and
                  Paul H. J. Kelly and
                  Felix J. Herrmann and
                  Gerard J. Gorman},
  title        = {Architecture and Performance of Devito, a System for Automated Stencil
                  Computation},
  journal      = {{ACM} Trans. Math. Softw.},
  volume       = {46},
  number       = {1},
  pages        = {6:1--6:28},
  year         = {2020},
  url          = {https://doi.org/10.1145/3374916},
  doi          = {10.1145/3374916},
  bibsource    = {dblp computer science bibliography, https://dblp.org}
}
